\documentclass[aps,pra,nopacs,nokeys,superscriptaddress,11pt,twoside,notitlepage,a4paper]{revtex4-2}

\usepackage{graphicx,epic,eepic,epsfig,amsmath,latexsym,amssymb,verbatim,color,mathtools}

\usepackage{theorem}
\newtheorem{definition}{Definition}

\newtheorem{lemma}[definition]{Lemma}

\newtheorem{theorem}[definition]{Theorem}

\usepackage{tikz-cd}

\usepackage{diagbox}

\def\squareforqed{\hbox{\rlap{$\sqcap$}$\sqcup$}}
\def\qed{\ifmmode\squareforqed\else{\unskip\nobreak\hfil
\penalty50\hskip1em\null\nobreak\hfil\squareforqed
\parfillskip=0pt\finalhyphendemerits=0\endgraf}\fi}
\def\endenv{\ifmmode\;\else{\unskip\nobreak\hfil
\penalty50\hskip1em\null\nobreak\hfil\;
\parfillskip=0pt\finalhyphendemerits=0\endgraf}\fi}
\newenvironment{proof}{\noindent \textbf{{Proof~}}}{\hfill$\blacksquare$}
\newenvironment{remark}{\noindent \textbf{{Remark~}}}{\qed}
\mathchardef\ordinarycolon\mathcode`\:
\mathcode`\:=\string"8000
\def\vcentcolon{\mathrel{\mathop\ordinarycolon}}
\begingroup \catcode`\:=\active
  \lowercase{\endgroup
  \let :\vcentcolon
  }

\newcommand{\nc}{\newcommand}
\nc{\rnc}{\renewcommand}
\nc{\lbar}[1]{\overline{#1}}
\nc{\bra}[1]{\langle#1|}
\nc{\ket}[1]{|#1\rangle}
\nc{\ketbra}[2]{|#1\rangle\!\langle#2|}
\nc{\braket}[2]{\langle#1|#2\rangle}
\nc{\proj}[1]{| #1\rangle\!\langle #1 |}
\nc{\avg}[1]{\langle#1\rangle}
\nc{\rank}{\operatorname{rank}}
\nc{\smfrac}[2]{\mbox{$\frac{#1}{#2}$}}
\nc{\tr}{\operatorname{Tr}}
\nc{\Tr}{\operatorname{Tr}}
\nc{\ox}{\otimes}
\nc{\dg}{\dagger}
\nc{\dn}{\downarrow}
\nc{\cA}{{\cal A}}
\nc{\cB}{{\cal B}}
\nc{\cC}{{\cal C}}
\nc{\cD}{{\cal D}}
\nc{\cE}{{\cal E}}
\nc{\cF}{{\cal F}}
\nc{\cG}{{\cal G}}
\nc{\cH}{{\cal H}}
\nc{\cI}{{\cal I}}
\nc{\cJ}{{\cal J}}
\nc{\cK}{{\cal K}}
\nc{\cL}{{\cal L}}
\nc{\cM}{{\cal M}}
\nc{\cN}{{\cal N}}
\nc{\cO}{{\cal O}}
\nc{\cP}{{\cal P}}
\nc{\cQ}{{\cal Q}}
\nc{\cR}{{\cal R}}
\nc{\cS}{{\cal S}}
\nc{\cT}{{\cal T}}
\nc{\cU}{{\cal U}}
\nc{\cV}{{\cal V}}
\nc{\cW}{{\cal W}}
\nc{\cX}{{\cal X}}
\nc{\cY}{{\cal Y}}
\nc{\cZ}{{\cal Z}}
\nc{\csupp}{{\operatorname{csupp}}}
\nc{\qsupp}{{\operatorname{qsupp}}}
\nc{\var}{{\operatorname{var}}}
\nc{\Var}{{\operatorname{Var}}}
\nc{\rar}{\rightarrow}
\nc{\lrar}{\longrightarrow}
\nc{\polylog}{{\operatorname{polylog}}}
\nc{\1}{{\openone}}
\nc{\wt}{{\operatorname{wt}}}
\nc{\av}[1]{{\left\langle {#1} \right\rangle}}

\def\ph{\varphi}

\nc{\RR}{{{\mathbb R}}}
\nc{\CC}{{{\mathbb C}}}
\nc{\FF}{{{\mathbb F}}}
\nc{\NN}{{{\mathbb N}}}
\nc{\ZZ}{{{\mathbb Z}}}
\nc{\PP}{{{\mathbb P}}}
\nc{\QQ}{{{\mathbb Q}}}
\nc{\UU}{{{\mathbb U}}}
\nc{\EE}{{{\mathbb E}}}
\nc{\id}{{\operatorname{id}}}

\nc{\CHSH}{{\operatorname{CHSH}}}

\nc{\<}{\langle}
\rnc{\>}{\rangle}
\nc{\Hom}[2]{\mbox{Hom}(\CC^{#1},\CC^{#2})}
\nc{\rU}{\mbox{U}}

\nc{\ob}[1]{#1}

\nc{\SEP}{{\text{SEP}}}
\nc{\NS}{{\text{NS}}}
\nc{\LOCC}{{\text{LOCC}}}
\nc{\PPT}{{\text{PPT}}}
\nc{\EXT}{{\text{EXT}}}
\nc{\Sym}{{\operatorname{Sym}}}

\nc{\ERLO}{{E_{\text{r,LO}}}}
\nc{\ERLOCC}{{E_{\text{r,LOCC}}}}
\nc{\ERPPT}{{E_{\text{r,PPT}}}}
\nc{\ERLOCCinfty}{{E^{\infty}_{\text{r,LOCC}}}}
\nc{\Aram}{{\operatorname{\sf A}}}

\newlength{\blank}
\begin{document}

\title{Playing Bayesian games better with separable quantum states\protect\\ than with any classical correlation}

\date{10 July 2026}

\author{Yaqing Xy Wang}
\affiliation{Department Mathematik/Informatik--Abteilung Informatik, Universit\"at zu K\"oln, Albertus-Magnus-Platz, 50923 K\"oln, Germany}
\affiliation{Institute for Theoretical Physics, Universit\"at zu K\"oln, 50937 K\"oln, Germany}
\affiliation{Forschungszentrum J\"ulich, Institute of Quantum Control, Peter Gr\"unberg Institut (PGI-8), 
52425 J\"ulich, Germany}
\email{ywang51@smail.uni-koeln.de}

\author{Giannicola Scarpa}
\affiliation{Escuela T\'ecnica Superior de Ingenier\'ia de Sistemas Inform\'aticos, Universidad Polit\'ecnica de Madrid, Calle Alan Turing, 28031 Madrid, Spain}
\email{g.scarpa@upm.es}

\author{Andreas Winter}
\affiliation{Department Mathematik/Informatik--Abteilung Informatik, Universit\"at zu K\"oln, Albertus-Magnus-Platz, 50923 K\"oln, Germany}
\affiliation{ICREA {\&} Grup d'Informaci\'o Qu\`antica, Departament de F\'isica, Universitat Aut\`onoma de Barcelona, 08193 Bellaterra (BCN), Spain}
%\affiliation{ICREA---Instituci\'o Catalana de Recerca i Estudis Avan{\c{c}}ats, Pg.~Llu\'is Companys, 23, 08010 Barcelona, Spain}
\email{andreas.winter@uni-koeln.de}

\begin{abstract}
Bayesian games, also known as games of incomplete information, are a fruitful arena for exploring the impact of correlations on a set of independent agents (players) via the game equilibria to which they give rise. 
It was realised some time ago that quantum states shared between the players can lead to new and beneficial equilibria, compared to classical correlation. While until now examples of this effect required an entangled state, here we show that even separable states can create new, genuinely quantum equilibria in games, that are advantageous with respect to all classically correlated equilibria. This shows that non-classical correlations beyond entanglement are indeed a resource, even in otherwise entirely classical situations. Our result brings quantum advantage in games significantly 
closer to possible realisation. 
It also illuminates and differentiates the existing hierarchy of ``legitimate notions of equilibrium'' in Bayesian games. 
\end{abstract}

\maketitle

\section{Introduction}
\label{sec:intro}
Bell inequalities \cite{Bell:inequality,CHSH}, which in the present work we like to view as nonlocal games with a common payoff function for all players, have famously shown that entangled states provide correlations not reproducible by shared classical information (``local hidden variables'' \cite{EPR}). More precisely, every pure entangled state provides a quantum correlation advantage in some nonlocal game \cite{Gisin:pure,PopescuRohrlich:pure}, while on the other hand no separable -- and indeed even some very mixed entangled states -- can possibly achieve the same \cite{Werner:separable}. From this point of view, separable quantum correlation is no different from classical correlation between separate parties. This is enshrined in the resource framework of quantum entanglement, cf.~\cite{Horodecki4}. 

Here, however, we will show that the picture changes dramatically when extending the framework from the cooperative games of Bell type to competitive games, where each player has their own payoff function (utility), whose expectation they are trying to maximise individually and which objective is in general in conflict with the other players' objectives \cite{vonNeumann:minimax,vonNeumannMorgenstern,Nash:equilibrium}. Because of this, the crucial notion in games is that of \emph{(Nash) equilibrium}, a kind of local optimum where no player can improve their own payoff unilaterally while the other players adhere to their strategies. Such games, with inputs (``types'') to which the players have to respond (``actions'') just as the above games of Bell type, have been considered in game theory and economics since the work of Harsanyi in the 1960s \cite{Harsanyi:anf} under the heading of \emph{games with incomplete information} or \emph{Bayesian games}. On the other hand, Aumann \cite{Aumann:correlated,Aumann:more-correlated} showed the impact of correlation on the equilibrium structure of games %of complete information 
(see also the precursor \cite{Raiffa:PhD}), while Forges unified these two strands in her pioneering work on correlated equilibria in Bayesian games \cite{Forges:5-legitimate-defs,Forges:corr-eq-revisited}. These models developed in an economic context finally meet the requirements to describe nonlocal games in the sense of Bell, and in particular the study of optimal game play using classical or quantum correlations, or indeed no-signalling (``belief-invariant'') correlations. What is missing in the latter is the element of competition, but since the breakthrough works of La Mura \cite{LaMura:game+quantum} and Pappa \emph{et al.} \cite{Pappa:CHSH-game}, a slowly growing literature is dedicated to quantum and beyond-quantum correlated equilibria, their superiority to classical correlation, and the separation of the different classes of correlated equilibria. While \cite{BigGame} proposed a unified theoretical framework for different kinds of correlation advice in Bayesian games and explored potentially practical cases of quantum advantage, several noteworthy papers have studied particular games or classes of games with advantageous quantum correlated equilibria: 
Bolonek-Laso\'n \cite{Bolonek-Lason:3-player-quantum,Bolonek-Lason:2-player-quantum}, Groisman \emph{et al.}~\cite{Groisman-et-al:pseudotelepathy}, Abbott \emph{et al.}~\cite{AMP:corr-vs-comm}, Cerd\`a \cite{Miquel:BSc,MiquelAndreas}, among others.
%\textcolor{red}{(We'd like a complete bibliography of quantum correlation advantage here, evidently without mentioning anything to do with EWL...)} 
All of these examples of quantum correlated equilibria with social welfare superior to that attainable by classically correlated equilibria depend on sharing a highly entangled state among the players, and indeed the increased social welfare is itself a Bell inequality violation. 

In the present work, answering questions raised indirectly in \cite{BigGame} (see also \cite{Jabir:MSc}), we show by explicit construction that there are Bayesian games for which separable quantum states give rise to equilibria outside the set of (classically) correlated equilibria, and indeed that the quantum state can help the players achieve a larger social welfare than any classically correlated equilibrium. 
Since the behaviour originating from the separable state is necessarily local, this also gives us the first example of a communication equilibrium that is local (hence it could in principle be prepared using a suitable shared random variable) but does not correspond to any correlated equilibrium, adding to the list of legitimate yet subtly different notions of correlated equilibrium in Bayesian games \cite{Forges:5-legitimate-defs,Forges:corr-eq-revisited}. 

\medskip
The rest of the paper is structured as follows. 
In Section \ref{sec:games+equilibria}, we introduce the necessary notations, and review the mathematical definitions of games and their (correlated) equilibria; in Section \ref{sec:main}, we present the game construction and the main results; in Section \ref{sec:numerics}, we construct competitive games by modifying the simplest nonlocal games with quantum advantage, the CHSH game, the Peres-Mermin magic square game and the GHZ game, and optimise their correlated equilibria numerically; we conclude in Section \ref{sec:conclusions}. Appendix \ref{app:legitimate} presents our updated landscape of now eleven(!) different notions of correlated equilibrium in Bayesian games.

\bigskip\noindent
\textbf{General notations.}
For ease of reading, we follow certain conventions for sets, their elements and random variables, as well as tuples (strings) of symbols: 

An \emph{alphabet} is simply a set $\cX$ ($\cY$, $\cZ$, etc), usually finite, of elements, sometimes called letters $x, x', x_i\in\cX$, $y\in \cY$, etc. The cardinality (number of elements) of a set $\cX$ is denoted $|\cX|$. 

\emph{Random variables} are denoted by capital letters $X$, $Y$, $Z$ and so on, taking values $x$, $y$, $z$, etc in alphabets $\cX$, $\cY$, $\cZ$, etc. A random variable always comes with its own distribution, even if it is not spelled out explicitly; the probabilities of events are referred to by $\Pr\{X=x\}$, $\Pr\{X\neq Y\}$ and more generally $\Pr\{X\in\cE\}$ for an event $\cE\subset\cX$. If required, we denote the distribution of $\mathbb{\cX}$ by $p$, and then say that $X$ has distribution $p$, symbolically $X\sim p$, meaning $\Pr\{X=x\} = p(x)$; note that since we are in the discrete setting, we will not distinguish between the probability distribution proper and the so-called probability mass function (aka probability vector). 

The \emph{total variation distance} between distributions $p$ and $q$ on the same alphabet $\cX$ is a norm, defined as 
\[
  \frac12\|p-q\|_1 := \sup_{\cA\subset\cX} p(\cA) - q(\cA) = \frac12 \sum_{x\in\cX} |p(x)-q(x)|.
\]

Finally, for alphabets $\cX_1,\ldots,\cX_n$, the Cartesian product $\cX := \cX_1 \times \cdots \times \cX_n$ is the \emph{set of $n$-tuples (or strings)}, denoted $x = (x_1,\ldots,x_n)$, if there is no danger of confusion with a product of numbers, even more compactly as $x = x_1\ldots x_n$. For a subset $I\subset[n]=\{1,\ldots,n\}$, $\cX_I = \prod_{i\in I} \cX_i$ is the set of tuples/strings indexed by $I$, its elements being $x_I = (x_i:i\in I)$, the projections of $x\in\cX$ onto the coordinates $I$. In the special case $I=[n]\setminus i$, we employ the common abbreviation $x_{-i} = x_{[n]\setminus i}$ for the tuple of all coordinates of $x$ except the $i$-th.

\section{Bayesian games and hierarchy of correlated equilibria}
\label{sec:games+equilibria}
%\subsection{Bayesian Game}
The exposition in the present section is necessarily short, but we refer the reader to the comprehensive textbook \cite{MaschlerSiolanZamir:GameTheory} for all the necessary background on game theory, the review article \cite{Brunner-et-al:nonlocality} and the excellent book \cite{Scarani} for the distinction between classical and quantum correlated behaviours, and \cite{BigGame} for the formalism of correlated equilibria incorporating quantum and no-signalling advice. 

\begin{definition}[{Harsanyi~\cite{Harsanyi:anf}, see~also~\cite{MaschlerSiolanZamir:GameTheory}}]
A Bayesian game $G$, or game of incomplete information, is given by the following data: 
\begin{itemize}
    \item $n$ players labelled $i\in [n] = \{1,\ldots,n\}$;
    \item the set $\cT:= \cT_1 \times \cT_2 \times \cdots \times \cT_n$ of type profiles;
    \item the set $\cA := \cA_1 \times \cA_2 \times \cdots \times \cA_n$ of action profiles;
    \item a prior probability distribution $p$ on the type profiles $t = t_1\ldots t_n \in \cT$, making the type profile $T=T_1\ldots T_n \sim p$ a random variable;
    \item payoff functions $u_i : \cT \times \cA \to \mathbb{R}$ for each $i\in[n]$.
\end{itemize}
\end{definition}

We next describe the game play in \emph{extensive form} (which lends itself more fruitfully to the incorporation of advice than the so-called strategic form).
The game $G$ starts with player types being sampled from the distribution $p$ where each player $i$ receives only their own type $t_i$ and has no information on the other players' types (beyond the known correlation through $p$). Players then need to individually come up with their actions $a_i$ based on the information they know (their own types and all the above parameters of the game, plus potential advice as introduced later), at which point the payoffs $u_i(t,a)$ are determined. 
In the absence of other input (such as correlated random variables, quantum states, etc, to be discussed presently), a pure strategy of player $i$ amounts to a function $g_i:\cT_i\rightarrow\cA_i$, the set of which is denoted $\cA_i^{\cT_i}$. In this \emph{strategic form} the game is now a game of complete information, with the $i$-th player's payoff function given by $u_i(g) = \EE u_i(T,g(T))$, where $g=g_1\ldots g_n$ and $g(t) = g_1(t_1)\ldots g_n(t_n)$. A mixed strategy of player $i$ is now simply a random variable $G_i \in \cA_i^{\cT_i}$. In particular, Nash's theorem applies, guaranteeing the existence of an equilibrium in (independent) mixed strategies; this will be revisited briefly below. 

%This process can be mathematically represented by a map from the information known to them (for Bayesian games only their own types $t_i$) to an element of the set of actions accessible to them $a_i$, \emph{i.e.} $g_i : \cT_i \to \cA_i.$ When the function deterministically maps a type to an action, the strategy is called a pure strategy. When the play draws $g_i$ from a probability distribution of a set of such strategies, we call the strategy a mixed strategy. To model the probabilistic mixed strategies of the players, we can also introduce local random variables $\lambda$ with distributions $\Lambda_i(\lambda_i)$ for all the strategies a player can adopt, such that $g_i = g_{i,\lambda_i}.$ We denote the player's conditional probability distribution of player $a_i \in A_i$ while given the type $t_i \in T_i$ by $g_{i,A|T}(a_i \vert t_i).$
%When we specify the strategies for all players $\textbf{g} = (g_1,g_2, \dots,g_n)$, we have a solution for a game.

%The players are then rewarded for their actions based on the payoff functions and their types. When player i receives type $t_i$, adopts a strategy that yields $g_{i,A|T}(a_i \vert t_i)$, they are rewarded $u_i (\textbf{t},\textbf{a}),$ where $\textbf{t}$ denotes all sampled player types for the game instance and $\textbf{a}$ all players' actions.

In the settings considered in the sequel, each action $a_i$ is no longer going to be a simple (deterministic or random) function of $t_i$, but the action profile $a=a_1\ldots a_n$ still has a well-defined probability distribution conditional on $t$, denoted $Q(a|t)$. This suffices to make types and actions into jointly distributed random variables, 
\[
  \Pr\{T=t,A=a\} = p(t)Q(a|t). 
\]
This allows us to consider the expected utility of a player, as it tells us the average value of reward the player may expect in this scenario:
\[
  \langle u_i \rangle = \EE u_i(T,A). 
\]
%where $\textbf{t}_{-i}$ denote the types of all players except player $i$. We will continue to use ${}_{-i}$ as a subscript to denote quantities of all players except the $i$-th.

We now make the distinction between a full coordination game and a game of conflicting interests. The former is a game where all players' payoff functions are the same, which means that they always favour the same outcome game situation $(t,a)$. The latter is a game where the players' payoff functions differ for a non-empty set of type profiles, \emph{i.e.}, there exists a type profile $t$ and action profiles $a \neq a'$ such that for two players $i$ and $j$, $a$ maximises player $i$'s expected utility, while $a'$ maximises player $j$'s.

\emph{Social welfare} is an indicator commonly evaluated for game equilibria. In this paper, it is defined to be the average of player utilities:
\begin{equation}
\label{eq:social_welfare}
  \text{SW}(Q) = \frac1n \sum_{i=1}^n \EE u_i(T,A).
\end{equation}

\subsection{Behaviours (aka ``correlations'')}
The statistical entity that encapsulates the players' response to the types is the conditional probability distribution $Q$ we call \emph{behaviour} (sometimes also referred to as ``correlation'', which however has already too many distinct meanings): 
moving temporarily away from the types and actions of above, $n$ agents encounter inputs $r_i\in\cR_i$ and are expected to respond with outputs $s_i\in\cS_i$, so that given input $r=r_1\ldots r_n$, the output $s=s_1\ldots s_n$ is seen with probability $Q(s_1\dots s_n|r_1\dots r_n)$. In the context of games, we think of $Q$ as embodying `advice', in some way or another provided by a mediator. Mathematically, all we require at this stage is $Q(s|r) \geq 0$ for all $s$ and $r$, and 
$\sum_s Q(s|r) = 1$ for all $r$. 
We denote the family of all such behaviours on input set $\cR$ and output set $\cS$ as $\textbf{ALL}(\cS\vert \cR).$

\subsubsection{Belief-invariance}
Out of all possible behaviours, there are sub-categories with particularly desired qualities. Belief-invariant, also called non-signalling, behaviours are those where the distribution of the outputs $s_i$ (given $r_i$) reveals no information on any of the other players' inputs $r_j$ ($j\neq i$).
Concretely, we call $Q$ \emph{belief-invariant for a subset $I \subset [n]$} of agents compared to the rest $J = [n]\setminus I$, if 
\[
  \forall s_I \in \cS_I, r_I \in \cR_I, r_J, r_J' \in \cR_J \quad 
  \sum_{s_J \in \cS_J} Q(s_I,s_J | r_I,r_J) = \sum_{s_J \in \cS_J}Q(s_I,s_J | r_I,r_J'), 
\]
meaning that there is a well-defined marginal behaviour $Q_I(s_I|r_I) = \sum_{s_J \in \cS_J} Q(s_I,s_J | r_I,r_J)$ of the parties in $I$. We call $Q$ simply belief-invariant (or no-signalling) if it has this property for all subsets $I\subset[n]$.
We denote the set of belief-invariant behaviours as $\textbf{BINV}(\cS \vert \cR)$.

\subsubsection{Locality}
Another important aspect of behaviours is locality. A conditional probability distribution $Q$ is called \emph{local} (or more historically accurate, described by local hidden variables) if the players can locally produce their output given their input as well as a shared random variable $\lambda$ distributed independently of the inputs $r$ according to a probability law $\mu$: 
%where $\boldsymbol{\lambda} = (\lambda_1, \dots , \lambda_n)$ is a variable shared among all players sampled from a certain distribution $\Lambda(\boldsymbol{\lambda})$ and itself independent of the players' inputs $\textbf{r}$. Denoting the players' local conditional distributions by $q_{i,\boldsymbol{\lambda}}$,
\[
  Q(s \vert r) 
    = \sum_{\lambda} \mu(\lambda) \prod_i Q_i(s_i | r_i,\lambda),
\]
where $Q_i(s_i|r_i,\lambda)$ is the conditional probability distribution used by party $i$.

Almost by definition, all local correlations are also belief-invariant, while the opposite is not necessarily the case. We denote the set of local correlations $\textbf{LO}(\cS|\cR)$.

\subsubsection{Quantumness}
Based on the laws of quantum mechanics, we see that it is also possible to generate correlations via quantum states and measurements. For quantum systems, the native description of a measurement is a POVM (positive operator valued measure). Consider a quantum state $\rho$ acting on a Hilbert space $\cH = \cH_1 \otimes \dots \otimes \cH_n$ that is the tensor product of all $n$ players' Hilbert spaces. In the player subsystems, we can consider sets of POVMs labelled by the superscript $r_i$. Each element of the set is a positive semi-definite matrix in $\cH_i$ that acts on the individual state according to Born's rule: 
\[
  \Pr\{s_i | r_i, \rho\} = \Tr \rho M^{r_i}_{s_i}, 
\]
where 
\(
  M^{r_i} = \bigl( M^{r_i}_{s_i} : s_i\in\cS_i \bigr)
\)
is a POVM, \emph{i.e.} for all $s_i$, $M^{r_i}_{s_i} \geq 0$ and $\sum_{s_i \in \cS_i} M^{r_i}_{s_i} = \1$. 
%After the measurement, the state of player $i$ collapses to the post-measurement state $\rho_i'= \frac{ M^{r_i}_{s_i} \rho_i}{\text{Pr}(s_i \vert r_i, \rho_i)}.$
When the players each apply their POVM that depends on the input $r_i$ on their individual Hilbert space $\cH_i$, it results in the joint measurement $M^{r} = M^{r_1} \otimes \cdots \otimes M^{r_n}$ on the global Hilbert space $\cH$. Associated to this joint measurement is the conditional probability distribution of players outputting $s$ from the outcomes of the measurements, \emph{i.e.}
\[
  Q(s|r) = \Pr\{s|r,\rho\}
    = \Tr \rho (M^{r_1}_{s_1} \otimes \cdots \otimes M^{r_n}_{s_n}).
\]

As can be checked easily, any behaviour obtained from measurement on a quantum system in this way is belief-invariant. However, the opposite need not be true. Also, a quantum behaviour need not be local, even though every local behaviour is quantum. We denote the set of quantum correlations by $\textbf{Q}(\cS \vert \cR)$.

\subsection{Equilibria: from Nash to correlated}
We are interested in families of game solutions that are optimal and stably so for the players, in the sense that no player has an incentive to change their adopted strategy assuming that the others adhere to theirs. We call these game solutions equilibria of the game, and for the nomenclature and precise definitions of the different types of equilibrium we consider here we refer to \cite{BigGame}. 
For instance, a Nash equilibrium is given by independent random functions $G_i\in\cA_i^{\cT_i}$ such that the following holds for every player $i$ and every alternate function $g_i'\in\cA_i^{\cT_i}$: 
\[
  \EE u_i(T,G(T)) \geq \EE u_i(T,G_{-i}(T_{-i})g_i'(T_i)).
\]

%\[\begin{split}
%   \langle u_{i,t_i}(\mathbf{g}) \rangle = \mathbb{E}_{\mathbf{t}_{-i}|t_i, \mathbf{g}} u_i\big(\mathbf{t}, g_1(t_1), \dots, g_n(t_n)\big) &\ge \mathbb{E}_{\textbf{t}_{-i}\vert t_i, (\textbf{g}_{-i},g'_i)}u_i(\textbf{t},\textbf{g}_{-i}\left(\textbf{t}_{-i}),g'_i(t_i)\right) = \langle u_{i,t_i}(\textbf{g}_{-i},g'_i)\rangle, \\
%   &\forall \, i\in N, t_i \in T_i, g'_i \in A_i^{T_i}. 
%\end{split}\]

We briefly review now the types of equilibrium associated to advice embodied in behaviours or quantum states considered in the sequel.

\subsubsection{Communication equilibria}
It may happen that in a game of incomplete information, the players have access to a correlation resource that may give them additional information. Operationally, this can be manifested as there being a trusted referee, who privately communicates with all players and shares with them their part of the correlation. The referee takes each player's input $r_i$, which may be a function $r_i=f_i(t_i)$ of their type (so that the referee gathers $r=r_1\ldots r_n$), samples from the distribution $Q_0(s|r)$, and provides each player $i$ with $s_i$ privately. The players may then use their type and the output $s_i$ to compute their individual action $a_i = g_i(t_i,s_i)$. 
In general both $F_i\in\cR_i^{\cT_i}$ and $G_i\in\cA_i^{\cT_i\times\cS_i}$ can be jointly distributed random functions (independent from the other players'), which together with $Q_0$ make $T$, $R$, $S$ and $A$ into jointly distributed random variables that hence allow us to define the expected payoffs $\EE u_i(T,A)$. We call the collection $(F,G,Q_0)$ a \emph{communication equilibrium} if for every player $i$ and every alternate function pair $f_i$ and $g_i$ (giving rise to a different random variable of actions $A'$), 
\[
  \EE u_i(T,A) \geq \EE u_i(T,A'). 
\]
In \cite{BigGame} it is shown that the resulting behaviour $Q(a|t)$, together with the trivial identity functions as pre- and post-processing is automatically a communication equilibrium, which we call the ``canonical form''. The set of all canonical-form communication equilibria of $G$ is denoted $\text{Comm}(G) \subset \textbf{ALL}(\cA|\cT)$. 

If $Q_0$ above is belief-invariant, we speak of a \emph{belief-invariant equilibrium} and it is easy to see that its canonical form is also belief-invariant. We denote the set of all belief-invariant equilibria in canonical form $\text{BI}(G) \subset \textbf{BINV}(\cA|\cT)$.

\subsubsection{Correlated equilibria}
As a special subclass of communication equilibria, \emph{(classically) correlated equilibria} are obtained by restricting the correlation to be a shared random variable that is independent of the players' inputs, \emph{i.e.} $Q_0(s|r) = Q_0(s)$.
By definition, all correlated equilibria are belief-invariant, and they also have a canonical form where $s_i\in\cS_i := \cA_i^{\cT_i}$. In other words, $Q(s)$ describes a joint distribution over the local random functions $G_i:\cT_i \rightarrow \cA_i$. 

The resulting behaviour $Q(a|t)$ of a correlated equilibrium is necessarily a local correlation, 
and we denote the set of these behaviours $\text{Corr}(G) \subset \textbf{LO}(\cA|\cT)$.

\subsubsection{Nash equilibria}
\emph{Nash equilibria} are obtained by further restricting $Q_0(s)$ to being a product distribution, $Q_0(s) = Q_1(s_1) \cdots Q_n(s_n)$. The resulting behaviour inherits this product form: \begin{equation}
  \label{eq:product-behaviour}
  Q(a|t) = Q_1(a_1|t_1) \cdots Q_n(a_n|t_n), 
\end{equation}
and we denote the set of behaviours of Nash equilibria $\text{Nash}(G) \subset \text{Corr}(G)$. As a matter of fact, $\text{Nash}(G)$ equals the intersection of $\text{Corr}(G)$ with the set of product behaviours as in \eqref{eq:product-behaviour}.

\subsubsection{Quantum correlated equilibria}
Directly going to the canonical form (cf.~\cite{BigGame}), a \emph{quantum correlated equilibrium} is given by an $n$-partite Hilbert space $\cH = \cH_1\otimes \cdots \otimes \cH_n$, a state $\rho$ of $\cH$ and POVMs $M^{t_i} = (M^{t_i}_{a_i} : a_i\in\cA_i)$ acting on $\cH_i$, such that for every player $i$ and every collection of alternate POVMs $M^{\prime t_i}$ the following holds: 
\[
  \EE u_i(T,A) \geq \EE u_i(T,A_{-i}A_i'),
\]
where 
\begin{align*}
  \Pr\{T=t,A=a\} 
    &= \Tr \rho(M^{t_1}_{a_1} \otimes \dots \otimes M^{t_n}_{a_n}), \\
  \Pr\{T=t,A_{-i}=a_{-i},A_i'=a_i\}
    &= \Tr \rho(M^{t_1}_{a_1} \otimes \cdots \otimes M^{t_{i-1}}_{a_{i-1}} \otimes M^{\prime t_i}_{a_i} \otimes \cdots \otimes M^{t_n}_{a_n}).
\end{align*}

%Suppose the players identify the inputs to the quantum correlation device $r_i$ with their given types $t_i$, and the outputs of the correlation $s_i$ to their final action $a_i$, we can simplify the expression for a player's expected utility for a game $G$ to the following:
%\[
%\langle u_{i,t_i} (\textbf{M}^\textbf{t},\rho)\rangle = \sum_{\textbf{t}_{-i},\textbf{a}} P(\textbf{t}_{-i}\vert t_i)\Tr \rho (M^{t_1}_{a_1} \otimes \dots \otimes M^{t_n}_{a_n})u_i(\textbf{t},\textbf{a}).
%\]
%
%For a solution $(\textbf{M}^\textbf{t},\rho)$ to be a quantum correlated equilibrium of game G, the following inequalities must, again, be satisfied for all players with all possible types and for all (alternative) quantum measurements in the Hilbert space $\cH$.
%\[\begin{split}
%    \langle u_{i,t_i} (\textbf{M}^\textbf{t},\rho)\rangle &= \sum_{\textbf{t}_{-i},\textbf{a}} P(\textbf{t}_{-i}\vert t_i)\Tr \rho (M^{t_1}_{a_1} \otimes \dots \otimes M^{t_n}_{a_n})u_i(\textbf{t},\textbf{a})\\
%    &\ge \sum_{\textbf{t}_{-i},\textbf{a}} P(\textbf{t}_{-i}\vert t_i)\Tr \rho (M^{t_1}_{a_1} \otimes \dots \otimes M^{t_{i-1}}_{a_{i-1}} \otimes M'^{t_i}_{a_i} \otimes \dots \otimes M^{t_n}_{a_n})u_i(\textbf{t},\textbf{a}), \quad \forall \, M'^{t_i}
%\end{split}\]

The set of behaviours $Q(a|t) = \Tr \rho(M^{t_1}_{a_1} \otimes \dots \otimes M^{t_n}_{a_n})$ of quantum correlated equilibria of the game $G$ is denoted $\text{Qu}(G) \subset \textbf{Q}(\cA|\cT)$. 
We have the evident inclusions
\[
  \text{Nash}(G) \subset \text{Corr}(G) \subset \text{Qu}(G) \subset \text{BI}(G) \subset \text{Comm}(G).  
\]
As remarked already, Nash's theorem applies to $\text{Nash}(G)$, showing that it -- and hence all the above sets -- is nonempty. It is easy to see that $\text{Corr}(G)$, $\text{Qu}(G)$, $\text{BI}(G)$ and $\text{Comm}(G)$ are convex sets, and it is well-known that $\text{Corr}(G)$, $\text{BI}(G)$ and $\text{Comm}(G)$ are actually compact polytopes. 

We have given here the conceptually simplest definitions of the different notions of equilibrium; in \cite{BigGame}, more compact (and in particular: more efficiently testable) formulations of the equilibrium condition in each class are discussed.

\section{A class of games with quantum correlated equilibrium\protect\\ 
superior to any classically correlated equilibrium}
\label{sec:main}
We start with a general recipe for converting any non-local game (or more generally 
a Bell inequality) with independent queries to the players and quantum 
advantage into a Bayesian competitive game for the same players but 
with altered type and action sets. For simplicity we give it first for 
two players, but at the end of the section formalise it for arbitrary 
number of players. 

A two-player nonlocal (cooperative) game $G$ for us is given by sets $\cX$ 
and $\cY$ of queries (the traditional name for the types in this setting), 
a product distribution $p\times q$ on $\cX\times\cY$ making $X$ and $Y$ two
jointly distributed independent random variables, output (action) sets 
$\cA$ and $\cB$, and a common payoff function $V(a,b,x,y) \in \RR_{\geq 0}$ 
taking values in non-negative reals. We shall assume w.l.o.g.~that $p(x)q(y)>0$
for all pairs $(x,y)$.
A classical strategy for the game is a pair of random functions 
$A:\cX\rightarrow\cA$ and $B:\cY\rightarrow\cB$ with a joint distribution
(which allows them to be correlated). 
By slight abuse of notation, but without danger of confusion, we declare 
$A=A(X)$ and $B=B(Y)$ as the random variables of the outputs, so that the 
payoff is $\EE V(A,B,X,Y)$, and the maximum classical payoff, the Bell local 
hidden variable limit, is 
\begin{equation}
  \label{eq:beta-def}
  \beta(G) := \max_{A,B \text{ random} \atop \text{functions}} \EE V(A,B,X,Y).
\end{equation}
We are interested in games where there is a quantum advantage: for this 
purpose, a quantum strategy is given by a state (density matrix) $\rho$ 
on a bipartite system $\cU\otimes\cV$ and local measurements (POVMs) 
$A^x = (A^x_a:a\in\cA)$ on $\cU$ and $B^y = (B^y_b:b\in\cB)$ on $\cV$. 
This defines a behaviour 
\begin{equation}
  \label{eq:qu-behaviour}
  Q(a,b|x,y) = \Tr \rho(A^x_a \ox B^y_b), 
\end{equation}
i.e. a conditional probability distribution of outputs $A$ and $B$ 
conditional on $X=x$ and $Y=y$. This makes $X$, $Y$, $A$ and $B$ into 
jointly distributed random variables, 
\[
  \Pr\{A=a,B=b,X=x,Y=y\} = p(x)q(y) Q(a,b|x,y), 
\]
so that we get a quantum payoff $\tau := \EE V(A,B,X,Y)$, where the 
expectation is calculated with respect to the above distribution. Let us 
assume that our game is of the type that shows a quantum advantage, 
more precisely that the contemplated strategy attains $\tau > \beta(G)$. 
Furthermore, that Bob's POVMs $B^y=(B^y_b:b\in\cB)$ are optimal (i.e. 
achieves the maximum payoff) given that $\rho$ and the $A^x=(A^x_a:a\in\cA)$ 
are fixed. (If $\cV$ is a finite-dimensional Hilbert space, this is 
guaranteed by the continuity of the payoff function and compactness.) 

Now, we are almost ready to define our associated Bayesian game 
$\widetilde{G}^\Lambda$, for a tunable parameter $\Lambda>0$, which has 
type spaces $\cS=\{\ast\}$ for Alice (i.e. she has no or trivial type) and 
$\cT=\cY$ for Bob (with the prior type distribution $q$), and action spaces 
$\cA' = \cX\times\cA$ for Alice and $\cB' = \cB\times\cX$ for Bob. 
The idea is that Alice produces both the input $x$ and the output $a$ of 
the nonlocal game $G$, whereas Bob gets the input $y$ and has to produce 
the output $b$ as well as a guess $\hat{x}$ of Alice's $x$: their payoff 
functions both have a common term $V(a,b,x,y)$, but in addition they play 
a zero-sum game where Bob is rewarded for $\hat{x} = x$ while he is penalised 
for $\hat{x} \neq x$, which heavily incentivises Alice to play $X\sim p$ and 
Bob to venture the guess $\widehat{X} \sim p$. The zero-sum game is defined by Bob's payoff 
matrix $W$ (Alice's is simply the negative $-W$, or more precisely a multiple thereof),
\begin{equation}
  W_{x,\hat{x}} = \begin{cases}
                    (|\cX|-1)p(x)^{-2}        & \text{ if } \hat{x} = x, \\
                    -p(x)^{-1}p(\hat{x})^{-1} & \text{ if } \hat{x} \neq x,
                  \end{cases}
\end{equation}
which in matrix form is written compactly as $W = P^{-1} (|\cX|\1-J) P^{-1}$, with the diagonal matrix $P=\operatorname{diag}\left(p(x):x\in\cX\right)$ and $J$ the all-$1$ matrix. 
From this we can see that $W$ is positive semidefinite, with the unique left and 
right annihilating eigenvector $p=(p(x):x \in \cX)$: $Wp = 0$, $p^\top W = 0$.

\begin{lemma}
  \label{lemma:W-game}
  The zero-sum game with payoff matrix $-W$ for the row player (Alice) and $+W$ for the column player (Bob) has a unique equilibrium, in which both players follow the mixed strategy $p$.
\end{lemma}
\begin{proof}
Clearly, if Alice plays the mixed strategy $p$, since $p^\top W = 0$, regardless of Bob's strategy $p'$, both Alice and Bob get payoff $p^\top W p' =0$; in other words, Alice's optimal payoff is $\geq 0$. 
Likewise, if Bob plays the mixed strategy $p$, since $Wp=0$, regardless of Alice's strategy $p'$, both Alice and Bob get payoff ${p'}^\top W p =0$; in other words, Bob's optimal payoff is $\geq 0$, too, hence the value of the game is indeed $0$. 

To argue the uniqueness of either player's optimal strategy, assume that Bob plays $p'\neq p$, thus $\omega := Wp' \neq 0$. But as $\sum_x p_x \omega_x = p^\top \omega = {p}^\top Wp' = 0$, there must exist an $x'$ with $\omega_{x'} = \delta_{x'}^\top\omega < 0$, hence playing $\delta_{x'}$ Alice can make a positive gain. The argument for Alice playing $p'\neq p$ is similar; or else, Bob could simply mirror Alice and play $p'$, too, giving him a payoff ${p'}^\top Wp' > 0$: due to the positive semidefiniteness of $W$, the value is $\geq 0$, and it cannot be $=0$ since that is only possible for the eigenvector $p$, which we excluded explicitly. 

For later use we derive a lower bound on Bob's expected winnings if Alice plays $p'\neq p$, in terms of the total variation distance of $p'$ from $p$:
\begin{equation}
  \label{eq:Bobs-gain}
  \text{If}\quad \frac12\|p'-p\|_1 > \delta, 
  \quad\text{then there exists}\ 
  x'\in\cX\ \text{s.t.}\quad 
  {p'}^\top W\delta_{x'} > \frac{\delta}{\|p\|^2} \frac{|\cX|}{|\cX|-1},
\end{equation}
where $\|p\| := \max_x p(x)$ and $\delta_{x'}$ is the pure strategy of deterministically playing $x'$. To see this, recall that 
\[\begin{split}
  \delta < \frac12\|p'-p\|_1 
  &= \sum_x \frac12|p'(x)-p(x)| \\
  &= \sum_{x\in\cX_+} (p'(x)-p(x))
  = \sum_{x\in\cX_-} (p(x)-p'(x)),
\end{split}\]
with $\cX_+ = \{x : p'(x)>p(x)\}$ and $\cX_- = \{x : p'(x)<p(x)\}$. Note that by our assumption, both of these sets are nonempty and $|\cX_+|+|\cX_-| \leq |\cX|$. More precisely, there exists $x\in\cX_+$ with $p'(x)-p(x)>\frac{\delta}{|\cX_+|}$, and for all $x\in\cX_-$ we have $p'(x)-p(x) =: -\delta_x < 0$ with $\sum_{x\in\cX_-} \delta_x > \delta$. Thus, if we define $x' := \arg\max \frac{p'(x)}{p(x)}$, 
%and $x^\circ := \arg\min \frac{p'(x)}{p(x)}$, 
then 
\[
  \frac{p'(x')}{p(x')} 
    >    1+\frac{\delta}{p(x')|\cX_+|}
    \geq 1+\frac{\delta}{\|p\| |\cX_+|}.
  \quad
  \forall x\in\cX_-\ \ 
  \frac{p'(x)}{p(x)} = 1-\frac{\delta_x}{p(x)}
    \leq 1-\frac{\delta_x}{\|p\|}.
\]
Thus, with Alice playing $p'$ and Bob $\delta_{x'}$, his expected payoff is 
\[\begin{split}
  {p'}^\top W \delta_{x'} 
    &= \frac{1}{p(x')} \sum_{x\neq x'} \left( \frac{p'(x')}{p(x')} - \frac{p'(x)}{p(x)} \right) \\
    &\geq \frac{1}{p(x')} \sum_{x\in\cX_-} \left( \frac{p'(x')}{p(x')} - \frac{p'(x)}{p(x)} \right) \\
    &> \frac{1}{\|p\|} \left( \frac{\delta| \cX_-|}{\|p\| |\cX_+|}+\frac{\delta}{\|p\|} \right)
    \geq \frac{\delta}{\|p\|^2} \frac{|\cX|}{|\cX|-1},
\end{split}\]
concluding the argument.
%\textcolor{blue}{This is an improved bound and still needs to be used as such in the main theorem below...}
\end{proof}

\medskip
With these preparations, we can write down the payoff functions of Alice and Bob in $\widetilde{G}^\Lambda$:
% in the new game $\widetilde{G}^\Lambda$:
\begin{align}
  \label{eq:payoff-A}
  u_A(x,a, b,\hat{x}, y) &:= V(a,b,x,y) - 2\Lambda W_{x,\hat{x}}, \\
  \label{eq:payoff-B}
  u_B(x,a, b,\hat{x}, y) &:= V(a,b,x,y) + \Lambda W_{x,\hat{x}}.
\end{align}
Next we can give the separable quantum advice state we propose to the players 
to use. Note that Alice has only trivial type, hence her advice is classical 
and the following is actually a classical-quantum (cq-)state: 
\begin{equation}
  \label{eq:omega}
  \omega^{A'B'} = \sum_{x,a,\hat{x}} p(x) \proj{xa}^{A'} 
                                     \ox p(\hat{x}) \bigl(\Tr_A \rho(A^x_a\ox\1) \ox \proj{\hat{x}}\bigr)^{B'},
\end{equation}
together with the original POVMs $B^y = (B^y_b:b\in\cB)$ for Bob.
The state $\omega$ can be interpreted as the players using the -- necessarily 
entangled(!) -- state $\rho$, but Alice generates her own sample $x$ of the 
input according to the distribution $p$, measures $A^x$ and records the output 
$a$; this leaves Bob with a post-measurement state, and in addition he generates 
$\hat{x}$ independently according to $p$. Of course, ultimately $\omega$ is 
fully separable, so it can be prepared without entanglement or exchanging qubits between the distant parties, rather using 
only classical correlation and local quantum state preparations. 

Playing this advice results in the behaviour
\begin{equation}
  \widetilde{Q}((x,a),(b,\hat{x})|\ast,y) = p(x)p(\hat{x})Q(a,b|x,y), 
\end{equation}
with the $Q$ from Eq. \eqref{eq:qu-behaviour} above. The first observation is that if Alice
and Bob play this quantum advice, $\EE u_A(X,A,B,\widehat{X},Y) = \tau 
= \EE u_B(X,A,B,\widehat{X},Y)$, so the social welfare SW(here defined to be the average of player payoffs) 
has expectation $\tau$. 
Secondly, the condition for quantum correlated equilibrium is satisfied 
for Bob: indeed, $\EE W(X,\widehat{X}) = 0$ for any choice of 
distribution of $\widehat{X}$, since $p^\top W = 0$. On the other hand, 
$\EE V(A,B,X,Y)$ is maximised by the very POVMs $B^y$ as per our assumption. 
Thirdly, it might be intuitive that also Alice's equilibrium condition is 
satisfied: on the one hand, again $\EE W(X,\widehat{X}) = 0$ for any choice of 
distribution of $X$, since $Wp = 0$; however, depending on the game $G$ it 
might be possible that after seeing the advice $X=x$ and $A=a$, Alice could 
potentially increase $\EE V(A,B,X,Y)$ by offering a different pair $(x',a')$. 
We will explore this for the concrete CHSH game below (Section \ref{sec:numerics}), 
but to proceed here we shall specialise to the subclass 
of quantum pseudo-telepathy games: those are games where 
$V(a,b,x,y) \in \{0,1\}$ is a Boolean predicate characterising ``win'', 
for which $\beta(G)<1$ is the maximum classical winning probability, whereas there 
is a quantum strategy achieving $\tau = \tau(G) = 1$, which we will assume henceforth. 
With this it is clear that we are in an equilibrium: Alice cannot 
unilaterally change the fact that $\EE W(X,\widehat{X}) = 0$, while 
$V(a,b,x,y) \in [0;1]$ and hence the same for any expectation value; 
at the same time, the advice already achieves $\EE V(A,B,X,Y) = \tau = 1$.

\begin{theorem}
  \label{thm:G-tilde:quantum-equilibrium}
  If $G$ is a quantum pseudo-telepathy game, then for any $\Lambda>0$ 
  the game $\widetilde{G}^\Lambda$ defined above has a quantum correlated 
  equilibrium in the state $\omega$ and the POVMs $(B^y_b:b\in\cB)$ above, 
  which achieves social welfare $1$. 
  
  The corresponding behaviour $\widetilde{Q}$ defined above is not the 
  behaviour of any classically correlated equilibrium. 
  Furthermore, for every $\varepsilon>0$ there exists a $\Lambda(\varepsilon)$ 
  such that for all $\Lambda \geq \Lambda(\varepsilon)$, the maximum social 
  welfare over all classically correlated equilibria is $\leq \beta(G)+\varepsilon$.
\end{theorem}
\begin{proof}
The first part is proved already by the preceding discussion. 

For the second part, we consider a general supposed correlated equilibrium,
which consists of jointly distributed random variables $X$, $A$ and $F$, 
where the former two are hopefully clear, and $F:\cY\rightarrow\cB\times\cX$ 
is a random function. By way of contradiction let us assume that 
\[\begin{split}
  \widetilde{Q}((x,a),(b,\hat{x})|\ast,y) 
    &= p(x)p(\hat{x})\Tr\rho(A^x_a \ox B^y_b) \\
    &\stackrel{!}{=} 
       \Pr\{X=x, A=a, F(y)=(b,\hat{x})\}.  
\end{split}\]
Thus, for every value $y$, $X$ and $\widehat{X}$ are independent and indeed 
$\Pr\{X=x,\widehat{X}=\hat{x}|y\} = p(x)p(\hat{x})$. 
In particular, $X \sim p$ independently of $y$. We will show next that 
$X$ and $F$ cannot be independent, for assume the opposite by way of contradiction, then we could write the joint distribution of $X$, $A$ and $F$ as 
\[\begin{split}
  \Pr\{X=x,A=a,F=f\} &= p(x)\Pr\{F=f\}\Pr\{A=a|X=x,F=f\} \\
                     &=: p(x)\mu(f)\alpha_f(a|x),
\end{split}\]
and we would obtain $Q(a,b|x,y)$ as a local correlation with the hidden 
variable $F$: 
\begin{equation}
  Q(a,b|x,y) = \EE_F \bigl(\alpha_F(a|x)\delta_{b,F(y)}\bigr).
\end{equation}
But note that $Q$ attains the Tsirelson value $\tau=1$ of the game by construction, 
whereas a local correlation can only attain $\leq \beta(G) < 1$. 
This contradiction shows that indeed $X$ and $F$ have some dependency, in particular there is a function $f:\cY\rightarrow\cB\times\cX$ 
occurring with positive probability $\Pr\{F=f\} > 0$, such that conditional on $F=f$, the distribution of $X$ is different from $p$: 
$p_f(x) := \Pr\{X=x|F=f\} \neq p(x)$ for some $x$. But in this case, Bob can increase his expected payoff by sampling $\hat{x}$ from $p_f$ rather than following the advice (Lemma \ref{lemma:W-game}), and now we have a contradiction to the assumption that we had an equilibrium. 

For the third part, we make this reasoning quantitative. For this purpose, consider 
a correlated equilibrium for $\widetilde{G}^\Lambda$, which is given by jointly 
distributed random variables $X\in\cX$, $A\in\cA$ and $F:\cY\rightarrow\cB\times\cX$. 
Note that at an equilibrium, $\EE W_{X,\widehat{X}} \geq 0$, for if it were
negative, Bob could always improve to $\EE W_{X,\widehat{X}} = 0$ by playing 
$\widehat{X} \sim p$ independently of $F$. This means that in equilibrium, 
the social welfare is $\leq \EE V(A,B,X,Y)$, and actually smaller than $\EE V$
by the amount $\EE W_{X,\widehat{X}} \geq 0$ if the latter is positive. 
As before, denote for every function $f$ with $\mu(f) = \Pr\{F=f\} > 0$ the conditional 
distribution $p_f(x) = \Pr\{X=x|F=f\}$ of $X$ given that $F=f$. Now one of two cases must occur:
\begin{enumerate}
  \item either $\sum_f \Pr\{F=f\} \frac12\|p-p_f\|_1 > \varepsilon$; 
  \item or $\sum_f \Pr\{F=f\} \frac12\|p-p_f\|_1 \leq \varepsilon$. 
\end{enumerate}

%\noindent
\textit{Case 1:} For each $F=f$, and denoting $\varepsilon_f = \frac12\|p_f-p\|_1$, by sampling $\widehat{X} \sim \delta_{x_f'}$ for an appropriate point mass at $x_f' \in \cX$, Bob can make $\EE( W_{X,\widehat{X}}|F=f) > \Delta(\varepsilon_f) := \frac{\varepsilon_f}{\|p\|^2}\frac{|\cX|}{|\cX|-1}$, according to Lemma \ref{lemma:W-game}, in particular Eq.~\eqref{eq:Bobs-gain} in its proof. This implies $\EE W_{X,\widehat{X}} \geq \sum_f \Pr\{F=f\} \Delta(\varepsilon_f) > \Delta(\varepsilon)$.
Hence, if $\Lambda \geq \Lambda(\varepsilon) := \frac{1}{2\Delta(\varepsilon)} = \frac{\|p\|^2}{2\varepsilon}\left(1-\frac{1}{|\cX|}\right)$, we find $\EE \Lambda W_{X,\widehat{X}} > 1$ and so $\EE u_A < 0$ for Alice's payoff. However, this cannot actually happen, because Alice, simply playing $X \sim p$ instead, could make all contributions to her payoff non-negative by annihilating $\EE W_{X,\widehat{X}} = 0$, and this hence has to be the case at equilibrium. We conclude that there cannot be any equilibrium in case 1.

%\noindent
\textit{Case 2:} The assumption means that $X$ and $F$ are almost independent, indeed we can rephrase it as $\frac12\|\PP(X,F) - p\ox\PP(F)\|_1 \leq \varepsilon$ for the joint and marginal 
distributions of $X$ and $F$. Thus, leaving $F$ alone and defining new random 
variables $X'$ and $A'$ with the joint distribution
\[
  \Pr\{X'=x,A'=a,F=f\} = p(x)\mu(f)\Pr\{A=a|X=x,F=f\},
\]
we conclude 
\(
  \frac12\|\PP(X,A,F) - \PP(X',A',F)\|_1 \leq \varepsilon, 
\)
and hence (since $0 \leq V \leq 1$)
\[
  \left| \EE V(A,B,X,Y) - \EE V(A',B,X',Y) \right| \leq \varepsilon.
\]
At the same time, since $X' \sim p$ and $A',B$ conditional on $X',Y$ is a local behaviour 
(as before, we use $F$ as the hidden variable), it must be the case that 
$\EE V(A',B,X',Y) \leq \beta(G)$. Thus, 
\[
  \EE(u_A+u_B) = 2\text{SW} = 2\EE V(A,B,X,Y) - \Lambda\EE W_{X,\widehat{X}} \leq 2(\beta(G)+\varepsilon), 
\]
recalling that due to equilibrium, $\EE W_{X,\widehat{X}} \geq 0$. 
%This concludes the proof.
\end{proof}

\medskip
\begin{remark}
One way to understand our construction and proof is to regard the zero-sum $W$ game as acting as a \emph{mechanism} \cite{Mechanism} in the given non-local game $G$: it provides a strong incentive for Alice to generate $X$ according to $p$ and independently of Bob (this is true both for the classical and the quantum correlation advice). In plain words, it functions as a (big) fine for Alice for revealing $X$, accompanied by a (relatively small) reward for Bob for showing her up. This ensures that her behaviour is essentially as if she were playing the nonlocal game $G$.

The difference between quantum and classical advice is that the former allows Bob to be informed about his ideal action (towards winning in the nonlocal game $G$) without revealing anything about $X$, whereas with the latter the same is impossible to realise unless Alice and Bob sacrifice a large part of their joint payoff. This could be called the \emph{price of privacy}, or indeed the price of knowing too much. For an instance of this principle in a different context cf.~\cite{Twain}. 
%\textcolor{blue}{Is here the place to discuss rock-paper-scissors with penalty for draw? (Because of the utility of correlation that hides information about one player's actions from the other players...)}
\end{remark}

%\medskip
%\begin{remark}
%Since we start from a quantum pseudo-telepathy game, the optimal quantum strategy is realised with a pure state $\rho = \proj{\varphi}$ (though not every optimal strategy requires that). By Naimark's dilation theorem, we can furthermore assume that Alice's measurements $A^x$ are projective, i.e.~all $A^x_a$ are projection operators. Then, we can exhibit a whole family of pure entangled states, all of which are quantum correlated equilibria and give rise to the same behaviour $\widetilde{Q}$: 
%\[ ... \]
%This is perhaps not so remarkable, except that the average over independent phases ... in the above state precisely recovers $\omega$.
%\end{remark}

%\medskip
%\textcolor{blue}{What happens when $\Lambda < \Lambda(\varepsilon)$ in the above proof? There might be a regime where we can make $\EE u_A$ positive but arbitrarily small and still have an equilibrium; this would be compensated by Bob's $\EE u_B > 1$, although the social welfare would remain bounded away from $2$...}

\medskip
\begin{remark}
Let us go back to the behaviour $\widetilde{Q}$ appearing in the proof, 
obtained from measuring the separable advice state $\omega$. As such, 
it must be a local correlation in the sense of Bell \cite{Brunner-et-al:nonlocality}, 
and at the same time it is a belief-invariant communication equilibrium 
\cite{BigGame} (cf. \cite{Forges:5-legitimate-defs}). However, the above argument shows that the same $\widetilde{Q}$ is not the behaviour of any 
classically correlated equilibrium. In other words, for the games $\widetilde{G}^\Lambda$ 
the set of behaviours of classically correlated equilibria is a strict 
subset of the intersection of belief-invariant communication equilibria 
with local behaviours: $\text{Corr}(\widetilde{G}^\Lambda) \subsetneq \text{BI}(\widetilde{G}^\Lambda) \cap \textbf{LO}$. 
%\textcolor{blue}{Discuss relation with Forges' legitimate notions of correlated equilibrium \cite{Forges:5-legitimate-defs,Forges:corr-eq-revisited}; 

This is analogous to \cite{AMP:corr-vs-comm}, where it was shown that there are belief-invariant equilibria whose behaviour is quantum, yet it is not the behaviour of a quantum correlated equilibrium: $\text{Qu}(G_{\text{AMP}}) \subsetneq \text{BI}(G_{\text{AMP}}) \cap \textbf{Q}$.

The reason why these inequalities are even possible comes from the subtle difference in the equilibrium conditions for belief-invariant and classically correlated advice: the former is expressed in terms of the behaviour $Q$ itself, the latter instead, assuming $Q$ is local, is a property of the joint distribution of local functions used to express $Q$. Not only is this distribution typically not unique to $Q$, but what makes the above constructions work is that these random functions reveal information to the players about each other that destroys the beneficial communication equilibrium.
\end{remark}

\medskip
Reflecting on the construction, we can see that the latter separation 
(between local communication equilibria and those coming from correlated equilibria) 
can be obtained directly from a Bell inequality with no-signalling advantage. 
Denote the no-signalling value of the game $G$ by $\nu(G)$ and assume that it is 
$>\beta(G)$ -- for example, the CHSH game has $\beta=\frac34$ and $\nu=1$ --, 
then we can construct the game $\widetilde{G}^\Lambda$ and the communication 
advice $\widetilde{Q}$ as above. The latter is clearly local (it has inputs only 
for one of the players), and if $V(a,b,x,y)\in\{0,1\}$ and $\nu(G)=1$, it is a 
communication equilibrium by the same reasoning as before the statement of 
Theorem \ref{thm:G-tilde:quantum-equilibrium}. The rest of the argument is identical, and we obtain the following.
\begin{theorem}
  \label{thm:G-tilde:local-comm-equilibrium}
  If $G$ is a no-signalling pseudo-telepathy game, then for any $\Lambda>0$ 
  the game $\widetilde{G}^\Lambda$ defined above has a communication equilibrium 
  $\widetilde{Q}$, which is a local behaviour and achieves social welfare $1$. 
  
  At the same time, this behaviour $\widetilde{Q}$ is not the 
  behaviour of any classically correlated equilibrium. 
  Furthermore, for every $\varepsilon>0$ there exists a $\Lambda(\varepsilon)$ 
  such that for all $\Lambda \geq \Lambda(\varepsilon)$, the maximum social 
  welfare over all classically correlated equilibria is $\leq \beta(G)+\varepsilon$.
  \hfill$\blacksquare$
\end{theorem}

%\textcolor{red}{Note that this can be expressed entirely in terms already existing in the mainstream game theory literature: correlated equilibrium, belief-invariantand possibly general communication equilibrium, ... We could start with this, even with the special case of CHSH and the Popescu-Rohrlich box \cite{PopescuRohrlich:no-signalling}. Then we proceed to show that some of these advantages can be achieved with quantum correlation without the need for a mediator processing the types.}

%\textcolor{green}{Btw, did you know that Tsirelson had worked on game theory, too? With Landsberger and Solan, and even on Bayesian games!}

\medskip
Finally, we make the game construction general for any number $n$ of players, and state the corresponding theorem about quantum and classical 
correlated equilibria (omitting the proofs, which are similar to the two-player case). 

We start from a nonlocal game of $n$ players with winning predicate $V$, where the settings $x_i$ are sampled independently from distributions $p_i$, \emph{i.e.} $p(x) = p_1(x_1)\cdots p_n(x_n)$. 
In the modified game $\widetilde{G}^\Lambda$ the payoff for each player is the sum of the original game payoff $V$ and a suitably scaled $W$ matrix payoff. Out of the $n$ players, $n-1$ are classical players, being advised to output $(x_i, a_i),$ where $\cX_i \ni x_i \sim p_i$. They all have the trivial type $\{*\}$. The $n$-th player is the quantum player with assigned, non-trivial type $x_n$. They receive the post-measurement quantum state with all the $n-1$ classical players' subsystems traced over. The quantum player proceeds to measure the state with the original quantum strategy's set of POVMs to generate their output $a_n$, as well as outputting a guess $\hat{x}_i$ at all other classical players' choice of $x_i$.
The payoff functions penalize all $n-1$ classical players for the quantum one's correct guesses, while rewarding them for the wrong ones. Accordingly, the quantum player is rewarded for the correct guesses and penalized for incorrect ones.

Formally, we define the $W$ matrix for $n$ players as the following: 
\[
W_{x_{-n},\hat{x}} = \begin{cases}
                    (\prod_{i=1}^{n-1}\vert \cX_i \vert - 1) \prod_{i=1}^{n-1} p_i(x_i)^{-2}       & \text{ if } \hat{x} = x_{-n}, \\
                    -\prod_{i=1}^{n-1} p_i(x_i)^{-1}p_i(\hat{x}_i)^{-1} & \text{ if } \hat{x} \neq x_{-n},
                  \end{cases}
\]
where $x_{-n}$ denotes the $n-1$ classical players' inputs to the original nonlocal game, $x_1\ldots  x_{n-1}$ and $\hat{x} = \hat{x}_1 \ldots  \hat{x}_{n-1}$ the quantum player's guess. 
The payoffs are then defined to be
\[
\begin{split}
    u_i &=  V(x,a) - \frac{2}{n-1}\Lambda W_{x_{-n},\hat{x}}\\
    u_n &=  V(x,a) + \Lambda W_{x_{-n},\hat{x}}\\
\end{split}
\]

If we imagine grouping the $n-1$ classical players' together, we see that the above described $n$- player game modification is in essence identical to the scenario where only one classical player is playing all $n-1$ outputs, making the argument and proofs from the two-player scenario equally applicable to the multi-player scenario. This results in the following analogues of Theorems \ref{thm:G-tilde:quantum-equilibrium} and \ref{thm:G-tilde:local-comm-equilibrium}:

\begin{theorem}
  \label{thm:G-tilde:quantum-equilibrium:multi}
  If $G$ is an $n$-player quantum pseudo-telepathy game, then for any $\Lambda>0$ 
  the game $\widetilde{G}^\Lambda$ defined above has a quantum correlated 
  equilibrium in the state $\omega$ and the POVMs $(B^y_b:b\in\cB)$ above, 
  which achieves social welfare $1$. 
  
  The corresponding behaviour $\widetilde{Q}$ defined above is not the 
  behaviour of any classically correlated equilibrium. 
  Furthermore, for every $\varepsilon>0$ there exists a $\Lambda(\varepsilon)$ 
  such that for all $\Lambda \geq \Lambda(\varepsilon)$, the maximum social 
  welfare over all classically correlated equilibria is $\leq \beta(G)+\varepsilon$.
  \hfill$\blacksquare$
\end{theorem}

\begin{theorem}
  \label{thm:G-tilde:local-comm-equilibrium-multi}
  If $G$ is an $n$-player no-signalling pseudo-telepathy game, then for any $\Lambda>0$ 
  the game $\widetilde{G}^\Lambda$ defined above has a communication equilibrium 
  $\widetilde{Q}$, which is a local behaviour and achieves social welfare $1$. 
  
  At the same time, this behaviour $\widetilde{Q}$ is not the 
  behaviour of any classically correlated equilibrium. 
  Furthermore, for every $\varepsilon>0$ there exists a $\Lambda(\varepsilon)$ 
  such that for all $\Lambda \geq \Lambda(\varepsilon)$, the maximum social 
  welfare over all classically correlated equilibria is $\leq \beta(G)+\varepsilon$.
  \hfill$\blacksquare$
\end{theorem}

\medskip
\begin{remark}
The application of the theorems is not restricted 
to using the optimal quantum or no-signalling strategy, nor indeed to starting with a pure state. In the quantum case, it is enough that Bob's measurements are ``locally'' optimal for the given state $\rho$, and the measurements that Alice and her sisters (Edith, Lorina, etc) make. This is enough to satisfy the equilibrium conditions for Bob -- for the Alices they amount to checking the classical, non-Bayesian conditions.

Similarly, in the no-signalling case, we could simply fix the overall behaviour that attains an advantage in the game $G$ and proceed from there, requiring only the equilibrium condition for this 
cooperative game for Bob.
\end{remark}

\section{Numerical evaluation of the quantum-vs-classical advantage}
\label{sec:numerics}
Here we investigate some explicit games coming from the construction in 
Section \ref{sec:main}, with two objectives: on the one hand, in the case 
that we start from a pseudo-telepathy game, we know already that the 
cq-state $\omega$ provides a quantum correlated equilibrium, and we 
can use linear programming \cite{Chvatal} to optimise the social welfare over 
classically correlated equilibria to show an explicit gap \cite{MaschlerSiolanZamir:GameTheory} (our argument in the preceding section is more of a proof of principle). Indeed, the conditions for a correlated equilibrium (in canonical form) of a Bayesian game form a finite list of linear inequalities in the probability distributions on $\cS = \cA_1^{\cT_1}\times\cdots\times\cA_n^{\cT_n}$.
On the other hand, for general Bell inequalities we would also need 
to check the quantum correlated equilibrium condition, which in the 
worst case boils down to checking finitely many cases (of possible 
deviations of Alice from her advice). In both scenarios we can additionally 
investigate if the quantum correlated equilibrium is stable in the sense 
that any nontrivial deviation from the advice leads to a strictly worse outcome for the player concerned. 

When presenting the results, we adhere to the following unified conventions.
First, with the construction of classical correlation as a resource, we calculate, for $\Lambda$ ranging from $0$ to $0.5$, the set of correlated equilibria for the modified game $\widetilde{G}^\Lambda$. Out of the classically correlated equilibria, we select, also for each $\Lambda$, the equilibria with: a) maximum original game(G) score $\EE V$, labelled ``Corr($\widetilde{G}^\Lambda$) max.V'' and b) maximum social welfare, as defined per the expression \eqref{eq:social_welfare}, labelled ``Corr($\widetilde{G}^\Lambda$) max.SW''. 

For both equilibria a) and b), their original game scores are compared to the classical bound of the game score $\beta(G)$ and the quantum bound of the original game, labelled ``Quantum(G)''. Then, we compare the individual utilities of the modified game for both equilibria, with labels specifying the party concerned ``$E[U\_*]$''. Finally, we compare the social welfare of the modified game for both equilibria and the reference bounds, where the equilibria induced by the separable quantum advice are labelled ``Quantum\_sep($\widetilde{G}^\Lambda$)''. Note that as argued in previous sections, the quantum separable equilibria achieve the social welfare identical to the quantum bound on the original game score.

For computationally feasible games, we also compare the correlated equilibria to Nash equilibria of the modified game. The Nash equivalent of a) and b) are labelled ``Nash (best V)'' and ``Nash (best SW)''.

\subsection{Magic square game}
The magic square (aka Peres-Mermin) game is a two-player quantum pseudo-telepathy game introduced over a series of papers by Asher Peres \cite{Peres:game}, David Mermin \cite{Mermin:game} and P. K. Aravind \cite{Aravind:magic-square}. 

The game features a hypothetical $3 \times 3$ table with entries $\pm 1$. Alice receives a row number as her input and must output entries for that row, while Bob receives a column number and similarly, fills out the column. To win the game, Alice's row must multiply to $+1$, and Bob's column to $-1$. Where Alice's row and Bob's column intersect, the entries reported must be consistent.

Classically, the game can be won with probability $\frac{8}{9}$ by two pre-filled tables for Alice and Bob, respectively. Of course, each entry where the two tables differ contributes a $\frac19$ chance of losing the game. It is easy to see that, no matter how they are constructed, there will be at least one row-column combination that does not satisfy the winning condition: indeed, if they could win with probability $1$, the two tables must be the same, but the product condition of Alice means that the product of all nine entries is $+1$ while the product condition for Bob means that the product is $-1$. This contradiction shows the claim, and it is easy to come up with tables that satisfy the product conditions of the rows/columns for Alice/Bob and differ in exactly one entry. 

With a quantum strategy, however, the game can be won with unit probability by exploiting the noncommutativity of quantum observables. 
Concretely, Alice and Bob could share the following quantum state:
$$|\psi\rangle = \frac{1}{2}(\vert 00\rangle + \vert11\rangle)_{A_1B_1} \otimes (\vert 00\rangle + \vert11\rangle)_{A_2B_2} ,
$$
where $\vert 0 \rangle$ and $\vert 1 \rangle$ are eigenstates of the Pauli matrix $\sigma_z$ with $+1$ and $-1$ eigenvalues respectively. Alice is given the qubits $A_1$ and $A_2$, while Bob gets $B_1$ and $B_2$.  
Upon receiving the row/column number, Alice/Bob chooses the corresponding one out of three bases. The resulting state from the measurement allows the final reply to be directly read out.
%A possible set of measurements for Alice and Bob is given below.
%\begin{align*} \label{ms_alice_meas}
%x&=0: \sigma_z \otimes \sigma_z   &  y&=0: -\sigma_x \otimes \sigma_z\\
%x&=1: \sigma_x \otimes \sigma_x   &  y&=1: -\sigma_z \otimes \sigma_x\\
%x&=2: \sigma_y \otimes \sigma_y   &  y&=2: \sigma_y \otimes \sigma_y
%\end{align*}
To retrieve the responses, one may refer to the following table, which in each row and column contains mutually commuting observables that the two parties measure on their respective qubit pairs: 
\begin{table}[ht]
%    \begin{center}
      \begin{tabular}{|r|r|r|}
        \hline
        $I \otimes Z$\ {} & $Z \otimes I$\ {} & $Z \otimes Z$\ {} \\
        \hline
        $X \otimes I$\ {} & $I \otimes X$\ {} & $X \otimes X$\ {} \\
        \hline
        $-X \otimes Z$\ {} & $-Z \otimes X$\ {} & $\phantom{-}Y \otimes Y$\ {} \\
        \hline
        \end{tabular}
%    \end{center}
    \caption{Mermin–Peres magic square: $X,Y,Z$ denote the Pauli matrices, $I$ the identity matrix.}
    \label{tab:ms_alg_table}
\end{table}

% \begin{align*}
% &\ket{\a_{++}}: \frac{\ket{00+01+10-11}}{2}  &  &\ket{\b_{++}}: \frac{\ket{00+11}}{\sqrt{2}}\\
% &\ket{\a_{+-}}: \frac{\ket{00-01+10+11}}{2}  &  &\ket{\b_{+-}}: \frac{\ket{01+10}}{\sqrt{2}}\\
% &\ket{\a_{-+}}: \frac{\ket{00+01-10+11}}{2}  &  &\ket{\b_{-+}}: \frac{\ket{00-11}}{\sqrt{2}}\\
% &\ket{\a_{--}}: \frac{\ket{00-01-10-11}}{2}  &  &\ket{\b_{--}}: \frac{\ket{01-10}}{\sqrt{2}}\\
% \end{align*}
The quantum behaviour possesses two key properties that enable the players to win with probability one for all possible type combinations (input row and column number from the referee). 
The product of each row in the Mermin-Peres square is $I \otimes I$, while for each column $-I \otimes I,$ meaning that the products of the row and column are respectively $1$ and $-1,$ as required to win the game. Secondly, one may check that for every $x$ and $y$, $\langle \psi\vert A^x_{a_{y}} \otimes B^y_{b_{x}} \vert \psi\rangle = 1,$ guaranteeing the cell entry where the row and column intersect is consistent between Alice and Bob. 

In the modified game $\widetilde{MS}^\Lambda$ the advice to the players consists of $(x,a)$ for Alice and $\Tr_{A_1A_2}(A^x_a \otimes \1)\proj{\psi} \otimes \proj{\hat{x}}$ for Bob, so that 
\[ 
  \omega^{A'B'} = \frac{1}{36} \sum_{\substack{x=0,1,2\\\hat{x}=0,1,2\\ a}} \proj{x,a}^{A'} \ox (\Tr_{A_1A_2}(A^x_a \otimes \1)\proj{\psi} \ox \proj{\hat{x}})^{B'},
\]
The players receive these payoffs:
\[
\begin{split}
    u_A &= V(x,a,y,b)- 2\Lambda W,\\
    u_B &= V(x,a,y,b)+ \Lambda W,
\end{split}
\]
where 
\begin{equation}
  W_{x,\hat{x}} = \begin{cases}
                    \phantom{-}18  & \text{ if } \hat{x} = x, \\
                    -9 & \text{ if } \hat{x} \neq x,
                  \end{cases}
\end{equation}

When Alice plays the suggested $(x,a)$ and Bob applies the same measurements as in the original game, they win the Magic Square game with probability 1. When Bob plays $\hat{x}$ also according to the advice, which is sampled uniformly, $\EE W = p^\top W p = 0.$ The overall expected utility for both Alice and Bob remains 1, leading to a social welfare of 1.

To compare this result to that of the correlated equilibria, we define the strategies considered. Alice and Bob receive the correlated advice of $(x,a,f),$ where $f:y \mapsto (b,\hat{x})$. There are, in total, $3\times 4 \times (4^3 \times 3^3) = 20736$ possible instances of advice, and we are optimising a probability distribution over that alphabet. 

Due to the immense size of the parameter space, we compute correlated equilibria for a limited sample of $\Lambda$s.
\begin{figure}[ht]
    \includegraphics[width=\linewidth]{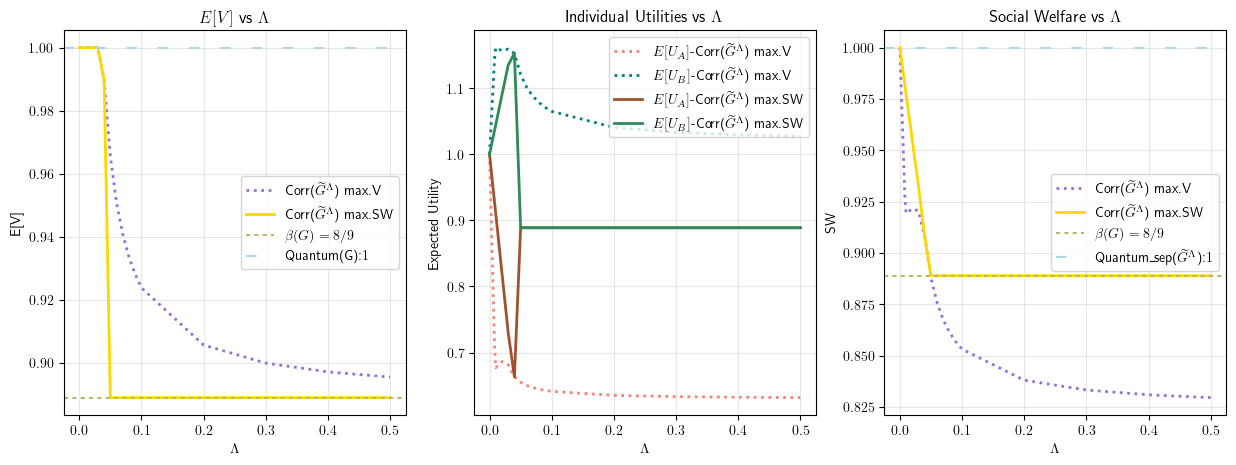}
    \makeatletter\long\def\@ifdim#1#2#3{#2}\makeatother
    \caption{\textbf{Left:} Magic square game score for correlated equilibria optimized for original game score and social welfare of modified game, all compared to classical bound of $\frac{8}{9}$ and the quantum bound of the original game score(also the perfect play score) $1$. \textbf{Middle:} Individual utility for players in correlated equilibria. Solid lines denote correlated equilibria that maximize the social welfare of the modified game, while dashed lines those that maximize the original game score. \textbf{Right:} Social welfare of correlated equilibria and quantum correlated equilibria. Refer to the beginning of Section \ref{sec:numerics} for the label convention.}
    \label{fig:MS_CE}
\end{figure}

\subsection{GHZ game}
The Greenberger-Horne-Zeilinger (GHZ) game is a three-player quantum pseudo-telepathy game originating 
in \cite{GHZ}. Unlike the original version, where the types are 
restricted to a subset satisfying a combinatorial constraint, here 
we describe it with independent uniform types; if the types do not 
satisfy the GHZ constraint, the players automatically win. 

In the original version of the GHZ game, there are three players, we call them Alice, Bob and Charlie. All three players play against a referee, who supplies binary inputs $(x,y,z)$ to the players, while the players respond with binary outputs $(a,b,c)$. Specifically to the original version, the input types are drawn uniformly from the following pool:
$$(x,y,z) \in \{(0,0,0), (1,1,0),(1,0,1),(0,1,1)\}.$$
The players win iff 
\[
  a\oplus b \oplus c 
    = x \vee y \vee z
    = \begin{cases}
        0 & \text{if } x=y=z=0, \\
        1 & \text{otherwise}.
      \end{cases}
\]  
In the present version of the GHZ game, the input types for players are drawn uniformly out of all possible binary combinations, 8 in total. When the input type is a member of the pool, the game proceeds as usual. When the input type is not a member of the pool, the players win the game automatically. This is necessary to have independent inputs so as to apply Theorem \ref{thm:G-tilde:quantum-equilibrium:multi}.

In the classical game, the players can discuss prior to the game and carry out a pre-communicated strategy. This yields a classical winning chance of $\frac{7}{8}$.
In the quantum strategy, the three players share a tripartite entangled state, the GHZ state:
$$\vert \psi\rangle = \frac{1}{\sqrt{2}}(\vert 000\rangle + \vert 111\rangle).$$
In case of their individual input type being $0$, they measure their bit in the $X$ basis and otherwise in the $Y$ basis. With this strategy, the players win the game with probability 1.
We apply the multi-player modification to the GHZ game, the resulting game $\widetilde{\text{GHZ}}^\Lambda$ has the payoff functions
\[\begin{split}
    u_A = u_B &=  V(x,a,y,b,z,c) - \Lambda W_{xy,\hat{x}\hat{y}}, \\
    u_C &= V(x,a,y,b,z,c) + \Lambda W_{xy,\hat{x}\hat{y}},\\
\end{split}\]
where both $xy$ and $\hat{x}\hat{y}$ range over four possibilities $\{0,1\}^2$. As all types $x,y,z$ are sampled uniformly in the nonlocal game, 
\begin{equation}
  W_{xy,\hat{x}\hat{y}} 
    = \begin{cases}
         \phantom{-}48  & \text{if } xy=\hat{x}\hat{y}, \\
                    -16 & \text{if } xy \neq\hat{x}\hat{y}.
       \end{cases}
\end{equation}
Alice and Bob receive uniformly distribited advice $(x,a)$ and $(y,b)$, respectively. At the same time, Charlie is given the advice 
\[
  \Tr_{AB}(A^x_a \otimes B^y_b \otimes \1)\proj{\psi} \otimes \proj{\hat{x},\hat{y}}.
\]
As usual, Charlie is also provided with the same set of POVMs to measure the given state as in the nonlocal game.
With the quantum strategy for the modified game $\widetilde{GHZ}^\Lambda$, all three players share the same expected utility of $1$, where the zero-sum $W$ game has annihilated payoff.

We now compute the correlated equilibria for $\widetilde{GHZ}^\Lambda$. The correlated advice for Alice, Bob and Charlie can be parametrized as $(x,a,y,b,f)$, where $f : z \mapsto (c, \hat{x}\hat{y})$.

\begin{figure}[ht]
%    {\centering
    \includegraphics[width=\linewidth]{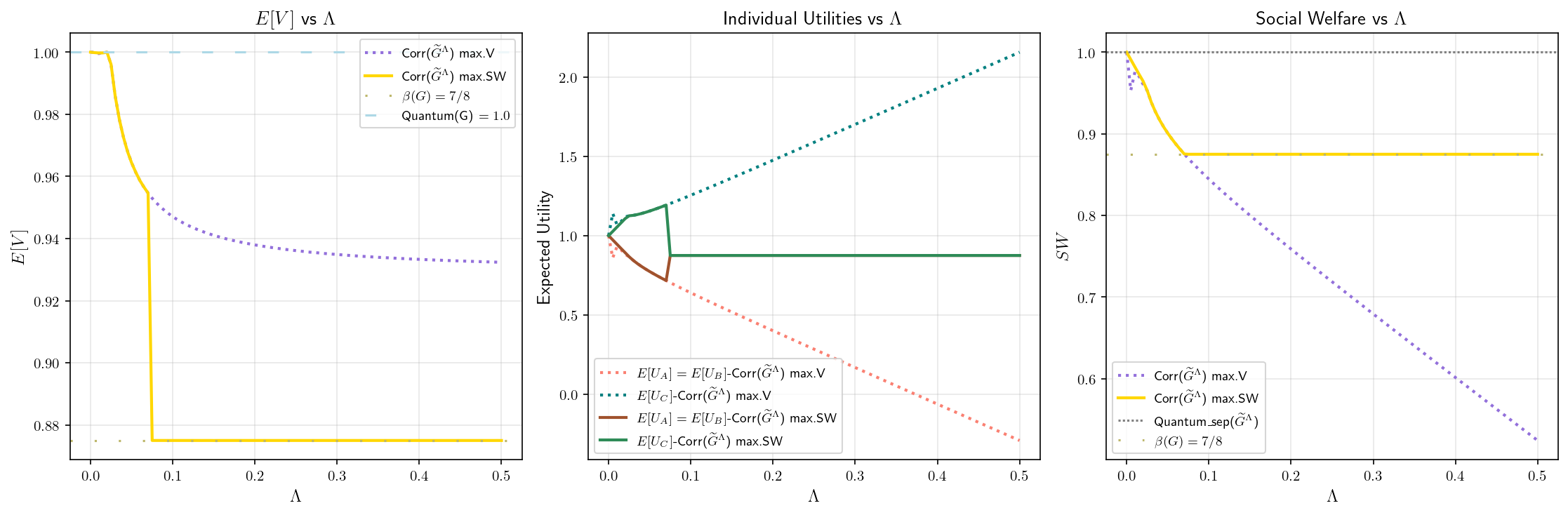}
%    }
    \caption{\textbf{Left:} Original GHZ game score for all players at correlated equilibria. Classical bound for the game is $\frac{7}{8},$ quantum bound is 1. Among the correlated equilibria, two objectives are maximized--Social welfare and GHZ score. \textbf{Middle:} Player utilities for correlated equilibria maximizing SW and GHZ score respectively. \textbf{Right:} Social welfare for correlated equilibria and quantum correlated equilibria. Refer to the beginning of Section \ref{sec:numerics} for the label convention.}
    \label{fig:GHZ_CE}
\end{figure}

\begin{figure}[ht]
    \centering
    \includegraphics[width=\linewidth]{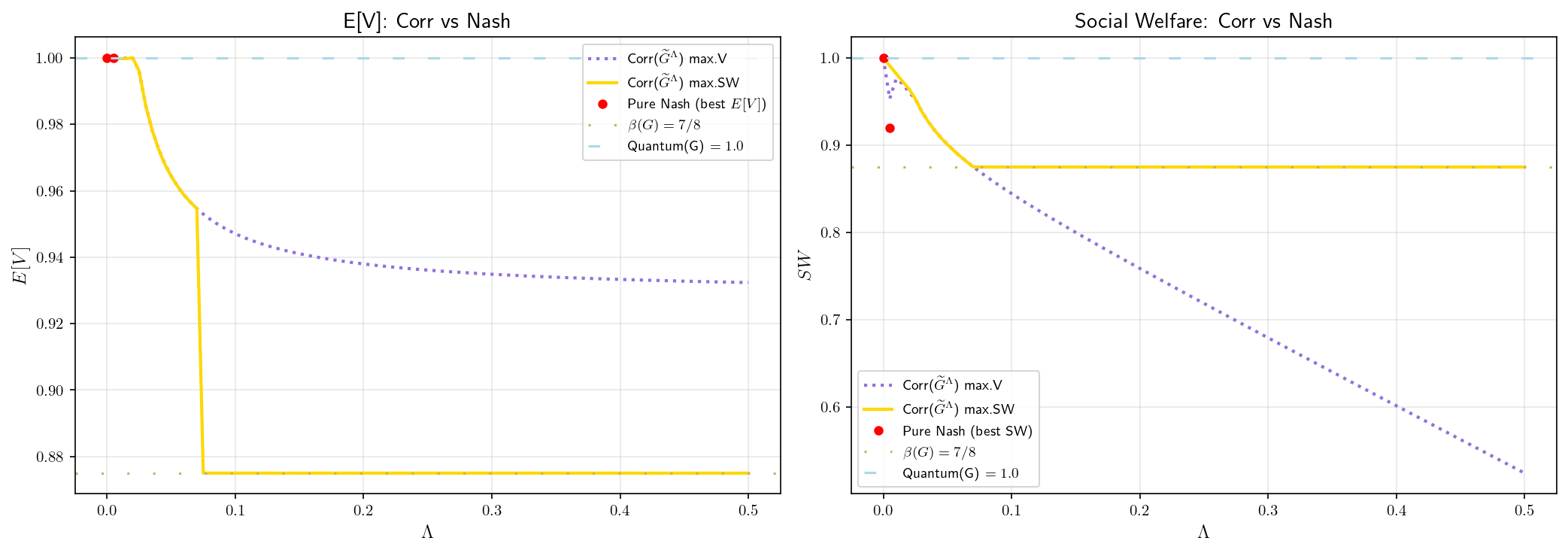}
    \caption{\textbf{Left:} GHZ score compared among Nash and correlated equilibria of the modified game. Two pure Nash equilibria exist around $\lambda = 0$, denoted with red dots. \textbf{Right:} Social welfare comparison.}
    \label{fig:GHZ_comparison}
\end{figure}

\subsection{CHSH game} 
%with quantum and no-signalling players}
The Clauser-Horne-Shimony-Holt (CHSH) game is a two-player, non-local cooperative game \cite{CHSH}. The players, Alice (A) and Bob (B) receive their inputs $x \in \cX$ and $y \in \cY$, and output their actions $a\in\cA$ and $b \in\cB$, where all inputs and outputs are binary: $\cA = \cB = \cX = \cY = \{0,1\}$. 
The types are sampled uniformly, \emph{i.e.}
\[
  p(xy) = \frac14 \text{ for all }x,y.
\]
The winning predicate is $V(a,b,x,y) = 1$ iff $a\oplus b = x\cdot y$.  

\medskip
Any pure or mixed classical strategy can win the game with a probability of at most $\frac34$. For example, the players can simply agree beforehand to align their output and only output 0 or 1 together, however, as only three out of four combinations of their possible types yield $x \cdot y = 0,$ their expected payoff is also $\frac34$ each:
\[
  \beta (\text{CHSH}) = \frac34.
\]

The optimal quantum strategy of the CHSH game famously beats the classical maximum expected payoff by achieving the Tsirelson bound:
\[
  \tau(\text{CHSH}) = \cos^2 \frac{\pi}{8} \approx 0.85.
\]
For this, Alice and Bob share a maximally entangled state $\ket{\Phi^+} = \frac{1}{\sqrt{2}}(\ket{00}+\ket{11})$ of two qubits, Alice's measurements are the Pauli $Z$ and $X$ observables, for $x=0,1$, respectively, while Bob's measurements are the rotated Pauli observables $Z'=\frac{1}{\sqrt{2}}(Z+X)$ and $X'=\frac{1}{\sqrt{2}}(Z-X)$ for $y=0,1$, respectively.

Finally, we describe the no-signalling strategy that wins the CHSH game with probability one, which is known as the Popescu-Rohrlich box:
\[
  \text{PR}(a,b|x,y) 
   = \begin{cases}
       \frac12 & \text{ if }\ a\oplus b = xy, \\
       0       & \text{ otherwise.}
     \end{cases}
\]
This attains the algebraic maximum for the no-signalling value, 
\[
  \nu(\text{CHSH}) = 1.
\]

\medskip
The modified CHSH game, where Alice now receives the trivial type and outputs $x \in \cX$ based on the advice $x \sim p$, Bob receives the binary type $y \in \cY$ and outputs both $b \in B$ as before and a guess at Alice's output $x$, $\hat{x}$.
The players' payoff are also modified according to Eqs.~\eqref{eq:payoff-A} and \eqref{eq:payoff-B}, where
\begin{equation}
  W_{x,\hat{x}} = \begin{cases}
                    \phantom{-}4  & \text{if }\ \hat{x} = x, \\
                    -4 & \text{if }\ \hat{x} \neq x,
                  \end{cases}
\end{equation}
or in matrix form
\[
  W_{x,\hat{x}} 
    = \begin{bmatrix}
        \phantom{-}4 & -4\\
       -4 & \phantom{-}4\\
      \end{bmatrix}.
\]

We want to argue that the advice state constructed in Eq.~\eqref{eq:omega} (and with the CHSH measurements for Bob) is an equilibrium of the modified game. Here it reads
\[ 
  \omega^{A'B'} = \frac18 \sum_{\substack{x,\hat{x}=0,1\\ a=0,1}} \proj{x,a}^{A'} \ox (\proj{\varphi_{a|x}} \ox \proj{\hat{x}})^{B'},
\]
where 
\[
  \ket{\varphi_{0|0}} = \ket{0},\ 
  \ket{\varphi_{1|0}} = \ket{1},\quad 
  \ket{\varphi_{0|1}} = \ket{+} := \frac{1}{\sqrt{2}}(\ket{0}+\ket{1}),\ 
  \ket{\varphi_{1|1}} = \ket{-} := \frac{1}{\sqrt{2}}(\ket{0}-\ket{1})
\] 
are the post-measurement states, which here happen to be the signal states of the famous BB84 protocol \cite{BB84}. 
In Section \ref{sec:main} we have already argued that the equilibrium conditions for Bob are satisfied, what is missing is to verify them for Alice.
Indeed, as argued before, when Bob follows the advice and makes his guess $\hat{x}$ according to $p(\hat{x})$, he is able to unilaterally annihilate the zero-sum game matrix $W_{x,\hat{x}}$. This means that Alice has no agency to increase the payoff she gets from $-\Lambda W_{x,\hat{x}}$. So it remains to check that, upon receiving the advice $\proj{x,a}$, Alice's optimal strategy is still to follow the advice and play $(x,a)$ accordingly. 

To verify this amounts to considering the four possible game scenarios characterised by $x$ and $a$, each binary. In each scenario, Bob receives the corresponding BB84 signal state, upon which he is advised to use the suggested CHSH POVMs. Depending on the input bit Bob gets, $y$, the POVMs $M^y_b = \proj{\Phi^y_b}$ read:
\begin{align*}
y = 0 &: \ket{\Phi^0_0} = (\cos\theta) \vert0\rangle + (\sin\theta)\vert 1\rangle,\ \ket{\Phi^0_1} = -(\sin\theta) \vert0\rangle + (\cos\theta)\vert 1\rangle, \\
y = 1 &: \ket{\Phi^1_0} = (\cos\theta) \vert0\rangle - (\sin\theta)\vert 1\rangle,\ \ket{\Phi^1_1} = (\sin\theta) \vert0\rangle + (\cos\theta)\vert 1\rangle.
\end{align*}

From here, one can calculate $P(b \vert y)$ for each $(x,a)$. As the original game only rewards winning cases with equal payoff of unity to both players, calculating Alice's expected payoff when she chooses to play $(x', a')$ is equivalent to summing the probabilities of winning cases.
\begin{figure}[ht]
%\centering
\begin{minipage}[t]{0.48\textwidth}
\centering
%\captionof{table}{$P(b|y)$}
\begin{tabular}{|c|c|c|}
\hline
\diagbox{a}{x} & \textbf{0} & \textbf{1} \\
\hline
\textbf{0} & 
\begin{tabular}{c|c|c}
\diagbox{b}{y} & 0 & 1 \\
\hline
0 & $c^2$ & $c^2$ \\
\hline
1 & $s^2$ & $s^2$
\end{tabular}
&
\begin{tabular}{c|c|c}

\diagbox{b}{y} & 0 & 1  \\
\hline
0 &  $\frac{(c + s)^2}{2}$ & $\frac{(c - s)^2}{2}$\\
\hline
1 &  $\frac{(c - s)^2}{2}$ & $\frac{(c + s)^2}{2}$
\end{tabular}
\\
\hline
\textbf{1} & 
\begin{tabular}{c|c|c}
\diagbox{b}{y} & 0 & 1 \\
\hline
0 & $s^2$ & $s^2$ \\
\hline
1 & $c^2$ & $c^2$
\end{tabular}
&
\begin{tabular}{c|c|c}
\diagbox{b}{y} & 0 & 1  \\
\hline
0 &  $\frac{(c - s)^2}{2}$ & $\frac{(c + s)^2}{2}$\\
\hline
1 &  $\frac{(c + s)^2}{2}$ & $\frac{(c - s)^2}{2}$
\end{tabular}
\\
\hline
\end{tabular}
\end{minipage}\hfill
\begin{minipage}[t]{0.48\textwidth}
\centering
%\captionof{table}{$V(a',x',b,y)$}
\begin{tabular}{|c|c|c|}
\hline
\diagbox{a}{x} & \textbf{0} & \textbf{1} \\
\hline
\textbf{0} & 
\begin{tabular}{c|c|c}
\diagbox{a'}{x'} & 0 & 1 \\
\hline
0 & $c^2$ & $\frac{1}{2}$ \\
\hline
1 & $s^2$ & $\frac{1}{2}$
\end{tabular}
&
\begin{tabular}{c|c|c}
\diagbox{a'}{x'} & 0 & 1  \\
\hline
0 &  $\frac{1}{2}$ & $\frac{(c + s)^2}{2}$\\
\hline
1 &  $\frac{1}{2}$ & $\frac{(c - s)^2}{2}$
\end{tabular}
\\
\hline
\textbf{1} & 
\begin{tabular}{c|c|c}
\diagbox{a'}{x'} & 0 & 1 \\
\hline
0 & $s^2$ & $\frac{1}{2}$\\
\hline
1 & $c^2$ & $\frac{1}{2}$
\end{tabular}
&
\begin{tabular}{c|c|c}
\diagbox{a'}{x'} & 0 & 1  \\
\hline
0 &  $\frac{1}{2}$& $\frac{(c - s)^2}{2}$\\
\hline
1 & $\frac{1}{2}$ & $\frac{(c + s)^2}{2}$
\end{tabular}
\\
\hline
\end{tabular}
\end{minipage}
\caption{\textbf{Left:} $P(b|y)$ calculated for each $(x,a)$ by applying Bob's suggested set of POVMs on the corresponding $\proj{\ph_{a|x}}$. \textbf{Right:} $V(a',x',b,y)$, the expected payoff for Alice when she plays $(x',a')$ while the advice is $(x,a)$. This calculation utilizes the prior (uniform) distribution for Bob's input y: $\text{Pr}(y=0) = \text{Pr}(y=1) = \frac{1}{2}$. 
Notation: $c = \cos\theta$, $s = \sin\theta$, where $\theta = \frac{\pi}{8}$.}
\label{Tab: CHSH_new_payoff}
\end{figure}

We have thus verified that the new advice yields indeed the optimal strategy for Alice, and is hence a quantum correlated equilibrium for the modified game $\widetilde{\text{CHSH}}^\Lambda$.

\medskip
Finally, for the optimisation of correlated equilibria, we cannot rely on Theorem \ref{thm:G-tilde:quantum-equilibrium} as it is not applicable, rather have to do the linear programming from first principles. 
Correlated equilibria concern jointly distributed random variables $X$, $A$ and $F$, where $F:\cY\rightarrow\cB\times\cX$ is a random function that maps the input Bob gets to his suggested action pair $(b, \hat{x})$. For the modified game $\widetilde{\text{CHSH}}^\Lambda$, the parameter space has dimension $63 = 2 \times 2\times 4 \times 4 -1$. As each advice corresponds to a payoff, the process of finding the correlated equilibria is a linear programming problem of verifying the candidate behaviour is optimal for each player: \emph{i.e.} $(x,a,b,\hat{x})$ satisfies
\begin{align*}
\textbf{For Alice: }\ \forall x',a'\ \EE u_A(x,a,b,\hat{x}) &\geq \EE u_A(x',a',b, \hat{x}),\\
\textbf{For Bob: }\ \forall b',\hat{x}'\ \EE u_B(x,a,b,\hat{x}) &\geq \EE u_B(x,a,b', \hat{x}').
\end{align*}
We thus arrive at the following results in Fig.~\ref{fig:CHSH_CE}.
They show the existence of a critical $\Lambda$ above which the social welfare of the correlated equilibria falls to the classical limit $\beta(G)$. For the game $\widetilde{CHSH}^\Lambda$, the limit is $\Lambda \geq \frac{1}{8}$. This can be reasoned by observing that the maximum payoff from $V(a,b,x,y)$ is 1, while Alice plays $x=0$ deterministically, she can control the parity of the game unilaterally, and as long as she stays with one deterministic $a$, Bob only needs to correlate to win the original CHSH game, this gives an overall social welfare of $1-2\Lambda.$ However, as Bob can equally play exactly the same $\hat{x} = 0$, Alice continues to get penalized by $-8\Lambda$, as previously argued, Alice's payoff can not possibly drop below 0, this strategy is only valid up to $\Lambda = \frac{1}{8},$ at which value the social welfare of the correlated equilibria is $1 - \frac{2}{8} = \frac{3}{4} = \beta(G)$.

\begin{figure}[ht]
    \centering
    \includegraphics[width=\linewidth]{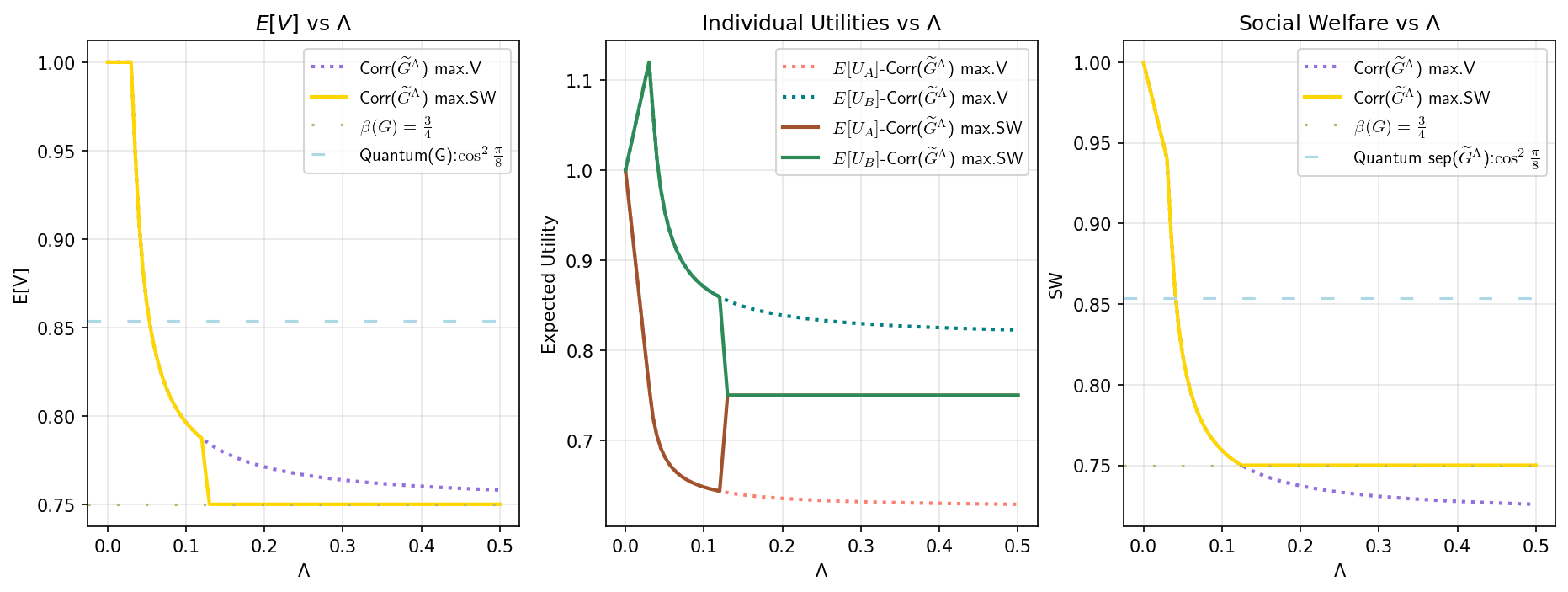}
    \caption{\textbf{Left:} CHSH score of $\widetilde{CHSH}^\Lambda$ correlated equilibria compared to $\beta$(CHSH)$=3/4$ and the quantum bound of $\cos^2\frac{\pi}{8}.$ \textbf{Middle: }expected utilities for Alice and Bob in correlated equilibria of $\widetilde{CHSH}^\Lambda$. \textbf{Right:} Social welfare comparison between equilibria that saturate social welfare and those that maximize the original game score. Refer to the beginning of Section \ref{sec:numerics} for the label convention.}
    \label{fig:CHSH_CE}
\end{figure}

We continue to compare the correlated equilibria of the modified game to its Nash equilibria. For Nash equilibria, the strategies for Alice and Bob are independent, Alice gets $(x,a)$ sampled from a distribution on $4$ points, and Bob his from distributions over functions $f: y \mapsto (b, \hat{x})$, for which there are $16$ possibilities. 
\begin{figure}[ht]
    \centering
     \includegraphics[width=\linewidth]{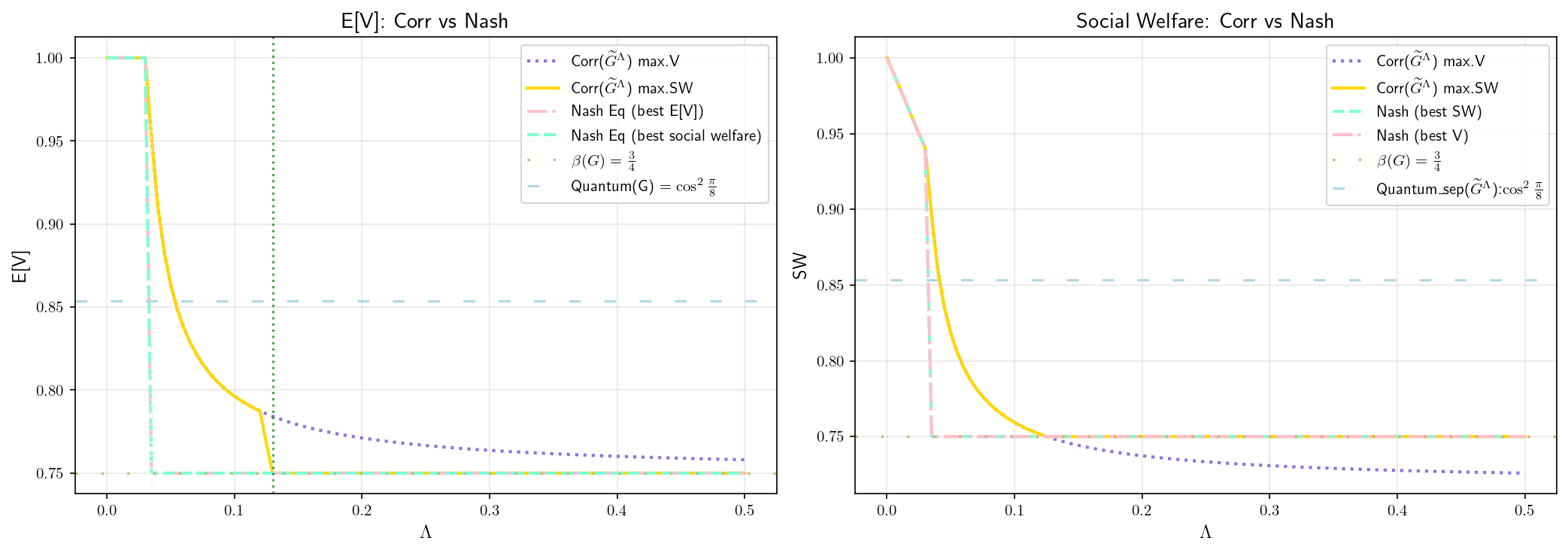}
    \caption{\textbf{Left: }CHSH score for correlated equilibria that saturate $\EE V(a,b,x,y),$ correlated equilibria that saturate social welfare, Nash equilibria that saturate $\EE V(a,b,x,y)$ and Nash equilibria with best social welfare. All four data series are juxtaposed with the classical limit $\frac{3}{4}.$ \textbf{Right: } Social welfare comparison between correlated equilibria and Nash equilibria of the modified game $\widetilde{CHSH}^\Lambda$.}
    \label{fig:CHSH_CE_Nash}
\end{figure}

\section{Conclusions}
\label{sec:conclusions}
We have found, for the first time, a quantum advantage in Bayesian games beyond entangled states, and indeed using barely non-classical cq-states. This marks a fundamental departure from all previous examples of game advantage due to quantum rather than classical correlation: those were all either directly Bell inequalities (cooperative games) or closely related to Bell inequalities in the sense that the types and actions were that of a nonlocal game, the payoff functions slight perturbations of the Bell parameter but such that their sum, the social welfare, was proportional to the original Bell parameter. Also our present examples each derive from a (very particular) nonlocal game, so cheating might be suspected, but the game (types, actions and payoffs) is modified significantly, so as to allow for a very different -- separable -- state to act as quantum advice.

The fact that quantum advantage in competitive games can be provided by separable quantum states, rather than requiring entanglement as in cooperative games, has manifold consequences: to start, it makes the potential realisation of quantum advantage much more accessible since separable states are much easier to manufacture and distribute than entangled states. Also, some of the games and the associated quantum advice are very simple (see the CHSH and GHZ examples), making it more plausible that a natural real-world application can be found. Thirdly, it raises the question of which quantum states can potentially offer a quantum (vs.~classical) advantage in a suitable Bayesian game; at the moment the only states that can be ruled out are the ``classical'' ones which are diagonal in a tensor product of local bases. 

A peculiar feature brought up by the numerical (linear programming) investigations of concrete games is that in those instances our Theorems \ref{thm:G-tilde:quantum-equilibrium} and \ref{thm:G-tilde:quantum-equilibrium:multi} miss something. Rather than the correlated equilibria of $\widetilde{G}^\Lambda$ having social welfare closer and closer to $\beta(G)$ with growing $\Lambda$, there is a cutoff: when $\Lambda\geq \Lambda_0$ the social welfare is simply $\leq \beta(G)$. It would be interesting to know if this is a general phenomenon for all games $G$. 

Further questions include: does every (maximum) Bell violation lead to a competitive game with quantum advantage via a separable state? Is perhaps steering already enough, noting that our instances of separable advice states encode steerable assemblages of quantum states \cite{WJD:steering,CavalcantiSkrzypzyk:steering,Uola-et-al:steering}?
Is it really just the presence of discord as Lowe \cite{Lowe:crazystuff} suggest? The latter seems far-fetched but both steering and non-zero discord are undoubtedly necessary conditions. For a cq-state like our $\omega$ encoding different assemblages, what we need is that one of its bipartite ``mother'' states and the corresponding measurements of Alice generating the assemblages give rise to a nonlocal correlation by complementing measurements of Bob. 
Finally, can we find a ``real-world'' game (\emph{i.e.} not purposely designed) where separable or other quantum states present a tangible quantum advantage in their equilibrium structure?

\acknowledgments 
%The authors are indebted to Hans-J\"urgen Buchner for consistently making the case for Bayesian rationality. 
GS and AW thank Jabir Thayyil for various discussions on different types of correlated equilibria in Bayesian games, which helped shape early forms of the ideas in the present paper. YXW thanks Mees Hendriks for delightful discussions and technical support.
GS was supported by the project
PID2023-146758NB-I00 funded by MICIU/AEI/10.13039/501100011033.
AW was supported by the European Commission QuantERA project ExTRaQT (Spanish MICIN grant no.~PCI2022-132965); by the Spanish MICIN (project PID2022-141283NB-I00) with the support of FEDER funds; 
by the Spanish MICIN with funding from European Union NextGenerationEU (PRTR-C17.I1) and the Generalitat de Catalunya; by the Spanish MTDFP through the QUANTUM ENIA project: Quantum Spain, funded by the European Union NextGenerationEU within the framework of the ``Digital Spain 2026 Agenda''; and by the Alexander von Humboldt Foundation.

%%%%%%%%%%%%%%%%%%%%%%%%%%%%%%%%%%%%%%%%%%%
%%% TO DO: Put references into bib file %%%
%%%%%%%%%%%%%%%%%%%%%%%%%%%%%%%%%%%%%%%%%%%

\vfill\pagebreak
\appendix

\section{Updates on the hierarchy of legitimate notions of\protect\\ correlated equilibrium in Bayesian games}
\label{app:legitimate}
Forges, in her articles on correlated equilibria in Bayesian games \cite{Forges:5-legitimate-defs,Forges:corr-eq-revisited}, traces four different formalisations of Bayesian games leading to equivalent notions of Nash equilibrium, but different (though related) concepts of correlated equilibrium, and proposes a fifth. 

%\textcolor{blue}{I (AW) find it hard to parse some of the distinctions made in the earlier paper (the later is better, and even better \cite{Lehrer-et-al:signalling}), but in our present framework we find eleven different notions, one of which corresponds to analysing the game in agent normal form; all of which collapse to the usual Aumann correlated equilibrium for games of complete information \cite{Zhang:games}.}

After considering carefully quantum advice, both separable and entangled, we now arrive at eleven(!) distinct notions of correlated equilibrium, which form partial hierarchies and which are separated from each other by examples from the previous literature or from our current paper. The mechanisms to distribute the correlation are different, but we present here unified definitions in terms of the game in extensive form, whereby the players have (free or regulated) access to distinct ``devices'' after learning their type and before having to announce their action. Each equilibrium gives rise to a joint behaviour $Q(a|t)$, and this allows us to compare, for a given game $G$, the eleven distinct sets of equilibrium behaviours, as convex subsets of $\textbf{ALL}(\cA|\cT)$.
This is a bit more refined than the attention to payoff vectors, which arise as linear functions of $Q$, and directly generalises the polytope of correlated equilibria according to Aumann in games of complete information, which are included as the case of trivial types $\cT_i = \{\ast\}$. It will be noted that in this case, all eleven sets of (behaviours of) correlated equilibria coincide, and yield the very polytope of correlated equilibria according to Aumann (see also \cite{Zhang:games}). 

\begin{figure}[ht]
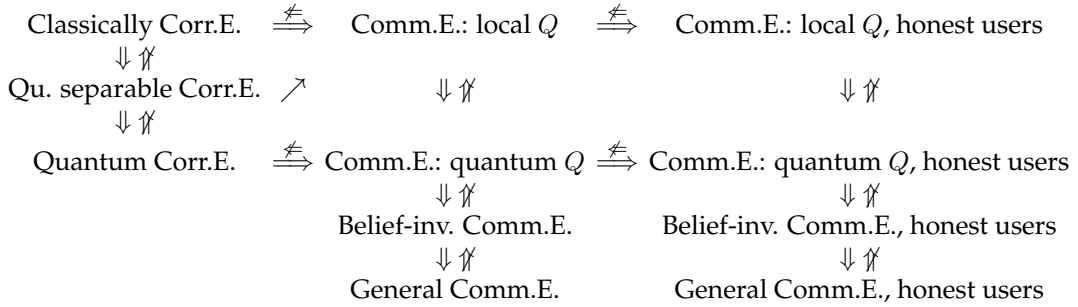

\begin{center}
\begin{tabular}{ccccc}
  Classically Corr.E. & $\stackrel{\not\Leftarrow}{\Longrightarrow}$ & Comm.E.: local $Q$ & $\stackrel{\not\Leftarrow}{\Longrightarrow}$ & Comm.E.: local $Q$, honest users \\
  $\Downarrow$ $\not\Uparrow$ & & & & \\
  Qu. separable Corr.E. & $\nearrow$ & $\Downarrow$ $\not\Uparrow$ & & $\Downarrow$ $\not\Uparrow$ \\
  $\Downarrow$ $\not\Uparrow$ & & & & \\
  Quantum Corr.E. & $\stackrel{\not\Leftarrow}{\Longrightarrow}$ & Comm.E.: quantum $Q$ & $\stackrel{\not\Leftarrow}{\Longrightarrow}$ & Comm.E.: quantum $Q$, honest users \\
  & & $\Downarrow$ $\not\Uparrow$ & & $\Downarrow$ $\not\Uparrow$ \\
  & & Belief-inv. Comm.E. & & Belief-inv. Comm.E., honest users \\
  & & $\Downarrow$ $\not\Uparrow$ & & $\Downarrow$ $\not\Uparrow$ \\
  & & General Comm.E. & & General Comm.E., honest users 
\end{tabular}
\end{center}
\caption{Eleven more or less legitimate notions of correlated equilibrium in Bayesian games.}
\label{fig:legitimate}
\end{figure}

Before we start, we give a little preview: we distinguish now three kinds of ``autonomous'' advice (classical, separable quantum, and general entangled quantum correlation), which are characterised by the feature that the advice can be given to the players even before they learn their types; four types of communication device (local, quantum, no-signalling, and general), which are accessed by the players as black boxes -- in particular players may query them with any input, not necessarily their individual type, and use the output as they see fit; and the same four types, but the players have to input their true type into the black box (alternatively modelled as an omniscient mediator who learns the players' types when they do), but as before they may use the output as they like. Within each of these three classes indicated by the letters A (autonomous), B (behaviours) and H (honest users), the advice is increasingly general, and an equilibrium in one class implies a corresponding equilibrium in the next class due to the relaxation of the equilibrium condition; 
see Fig. \ref{fig:legitimate}. 

\begin{itemize}
\item[\textbf{A1}] The set $\text{Corr}(G)$ of classically correlated equilibria, which are described by ordinary Aumann correlated equilibria in the strategic normal form of $G$, where each player's strategy space is simply $\cS_i = \cA_i^{\cT_i}$. Note however that here we only consider the  associated behaviour in $\textbf{LO}(\cA|\cT)$. All other notions of correlated equilibrium considered have to be described in the game's extensive form. 

\item[\textbf{A2}] The set of quantum correlated equilibria with a separable quantum state, which we might denote $\text{Qu}_{\text{SEP}}(G)$. It contains $\text{Corr}(G)$, the classically correlated equilibria, but can be strictly larger as we showed here. 

\item[\textbf{A3}] The set $\text{Qu}(G)$ of all quantum correlated equilibria is strictly larger than $\text{Qu}_{\text{SEP}}(G)$ as shown by Bell inequalities \cite{Bell:inequality,CHSH,GHZ}, or in competitive games by the examples of La Mura \cite{LaMura:game+quantum}, Pappa \emph{et al.} \cite{Pappa:CHSH-game}, and subsequent ones, cf. \cite{BigGame}. 

\item[\textbf{B1}] The set of communication equilibria with local behaviour, \emph{i.e.} $\text{Comm}(G) \cap \textbf{LO}(\cA|\cT)$. As observed here, it contains $\text{Corr}(G)$ but can be strictly larger. It also contains $\text{Qu}_{\text{SEP}}(G)$, but it is an open question whether it can be strictly larger. 

\item[\textbf{B2}] The set of communication equilibria with quantum behaviour, \emph{i.e.} $\text{Comm}(G) \cap \textbf{Q}(\cA|\cT)$. It contains $\text{Qu}(G)$ but can be strictly larger than that, as shown by Abbott \emph{et al.} \cite{AMP:corr-vs-comm}. Bell inequalities show that in general, $\text{Qu}(G)$ and hence also $\text{Comm}(G) \cap \textbf{Q}(\cA|\cT)$ contains points outside $\text{Comm}(G) \cap \textbf{LO}(\cA|\cT)$. 

\item[\textbf{B3}] The set of belief-invariant communication equilibria, $\text{BI}(G) = \text{Comm}(G) \cap \textbf{BINV}(\cA|\cT)$. This is typically larger than $\text{Comm}(G) \cap \textbf{Q}(\cA|\cT)$ via the violation of Tsirelson bounds by no-signalling correlations \cite{Tsirelson:bound,PopescuRohrlich:no-signalling}. It might be that Theorem \ref{thm:G-tilde:local-comm-equilibrium} provides examples of separation for competitive games if we could upper-bound the social welfare of $\widetilde{G}^\Lambda$ under quantum correlated equilibria by $\tau(G)+\varepsilon$.

\item[\textbf{B4}] The set of general communication equilibria, $\text{Comm}(G)$, where we now allow behaviours $Q$ that are potentially signalling between the players. 

\item[\textbf{H1}] The set of communication equilibria with local behaviour, but where the players honestly reveal their type to the communication device (equivalently, the mediator has prior knowledge of the types). Forges \cite{Forges:5-legitimate-defs,Forges:corr-eq-revisited} showed that this is in general strictly larger than $\text{Comm}(G) \cap \textbf{LO}(\cA|\cT)$. This concept of correlated equilibrium is equivalent to that arising from the \emph{agent normal form} of the game, which is a game of complete information but with a larger number of players, cf.~\cite{Harsanyi:anf,Forges:5-legitimate-defs}.

\item[\textbf{H2}] The set of communication equilibria with quantum behaviour, but where the players honestly reveal their type to the communication device. Again, Bell inequalities with quantum violation show that this can be strictly larger than the previous set (with local behaviours).

\item[\textbf{H3}] The set of communication equilibria with no-signalling behaviour, but where the players honestly reveal their type to the communication device. The examples in Forges \cite{Forges:5-legitimate-defs,Forges:corr-eq-revisited} show that this can be strictly larger than the analogue set with local behaviours. In fact, Tsirelson inequalities with no-signalling violation show a separation from the previous set (with quantum behaviours). 

\item[\textbf{H4}] The set of communication equilibria with general behaviour, but where the players honestly reveal their type to the communication device. Forges \cite{Forges:5-legitimate-defs,Forges:corr-eq-revisited} shows that this can be strictly larger than the previous set (with belief-invariant/no-signalling behaviours).

\end{itemize}

\end{document}